\documentclass[english,american,final,table]{IEEEtran}
\usepackage[T1]{fontenc}
\usepackage[latin9]{inputenc}
\usepackage{xcolor}
\usepackage{pdfcolmk}
\usepackage{array}
\usepackage{float}
\usepackage{mathrsfs}
\usepackage{bm}
\usepackage{amsmath}
\usepackage{amsthm}
\usepackage{amssymb}
\usepackage{graphicx}
\PassOptionsToPackage{normalem}{ulem}
\usepackage{ulem}

\makeatletter

\providecommand{\tabularnewline}{\\}
\floatstyle{ruled}
\newfloat{algorithm}{tbp}{loa}
\providecommand{\algorithmname}{Algorithm}
\floatname{algorithm}{\protect\algorithmname}
\providecolor{lyxadded}{rgb}{0,0,1}
\providecolor{lyxdeleted}{rgb}{1,0,0}

\DeclareRobustCommand{\lyxdeleted}[3]{{\color{lyxdeleted}\lyxsout{#3}}}
\DeclareRobustCommand{\lyxsout}[1]{\ifx\\#1\else\sout{#1}\fi}

\theoremstyle{plain}
\newtheorem{prop}{\protect\propositionname}
\theoremstyle{plain}
\newtheorem{thm}{\protect\theoremname}
\theoremstyle{plain}
\newtheorem{lem}{\protect\lemmaname}

\PassOptionsToPackage{dvipsnames,table}{xcolor}

\usepackage{amsmath,amssymb}

\usepackage{algorithm}
\usepackage{algorithmic}

\usepackage[table]{xcolor}

\usepackage{graphicx,psfrag,cite,subfigure}
\usepackage{epstopdf}

\author{
Yuanshuai~Zheng and Junting~Chen

\IEEEauthorblockA{School of Science and Engineering,
Shenzhen Future Network of Intelligence Institute (FNii-Shenzhen), and\\
Guangdong Provincial Key Laboratory of Future Networks of Intelligence\\
The Chinese University of Hong Kong, Shenzhen, Guangdong 518172, P.R. China}

\thanks{The work was supported in part by the National Science Foundation of China 
under Grant No. 62171398, by the Basic Research Project No. HZQB-KCZYZ-2021067 of Hetao Shenzhen-HK S\&T Cooperation Zone, 
by the Shenzhen Science and Technology Program under Grant No. JCYJ20210324134612033 and No. KQTD20200909114730003, 
by Guangdong Research Projects No. 2019QN01X895, No. 2017ZT07X152, and No. 2019CX01X104, 
by the Shenzhen Outstanding Talents Training Fund 202002, 
by the Guangdong Provincial Key Laboratory of Future Networks of Intelligence (Grant No. 2022B1212010001), 
by the National Key R\&D Program of China with grant No. 2018YFB1800800, and 
by the Key Area R\&D Program of Guangdong Province with grant No. 2018B030338001.
}

}


\usepackage[acronym]{glossaries}
\newcommand{\newac}{\newacronym}
\newcommand{\ac}{\gls}

\makeglossaries

\usepackage{enumitem}
\usepackage{cases}

\renewcommand{\lyxdeleted}[3]{{\color{lyxdeleted}{}}}

\newac{speb}{SPEB}{square position error bound}
\newac[plural=EFIMs,firstplural=Fisher information matrices (EFIMs)]{efim}{EFIM}{Fisher information matrix}
\newac{ne}{NE}{Nash equilibrium}
\newac{mse}{MSE}{mean squared error}
\newac{toa}{TOA}{time-of-arrival}
\newac{snr}{SNR}{signal-to-noise ratio}
\newac{lan}{LAN}{local area network}
\newac{psd}{PSD}{positive semidefinite}
\newac{pd}{PD}{positive definite}
\newac{wrt}{w.r.t.}{with respect to}
\newac{lhs}{L.H.S.}{left hand side}
\newac{wp1}{w.p.1}{with probability 1}
\newac{kkt}{KKT}{Karush-Kuhn-Tucker}
\newac{wlog}{w.l.o.g.}{without loss of generality}
\newac{mle}{MLE}{maximum likelihood estimation}
\newac{gps}{GPS}{global positioning system}
\newac{rssi}{RSSI}{received signal strength indicator}
\newac{mimo}{MIMO}{multiple-input multiple-output}
\newac{csi}{CSI}{channel state information}
\newac{fdd}{FDD}{frequency division duplexing}
\newac{ms}{MS}{mobile station}
\newac{bs}{BS}{base station}
\newac{d2d}{D2D}{device-to-device}
\newac{slnr}{SLNR}{signal-to-interference-leakage-and-noise-ratio}
\newac{ula}{ULA}{uniform linear antenna array}
\newac{pas}{PAS}{power angular spectrum}
\newac{mmse}{MMSE}{minimum mean square error}
\newac{zf}{ZF}{zero-forcing}
\newac{rzf}{RZF}{regularized zero-forcing}
\newac{as}{AS}{angular spread}
\newac{aod}{AOD}{angle of departure}
\newac{iid}{i.i.d.}{independent and identically distributed} 
\newac{sinr}{SINR}{signal-to-interference-and-noise ratio}
\newac{tdd}{TDD}{time-division duplex}
\newac{rvq}{RVQ}{random vector quantization}
\newac{rhs}{R.H.S.}{right hand side}
\newac{mrc}{MRC}{maximum ratio combining}
\newac{cdf}{CDF}{cumulative distribution function}
\newac{a.s.}{a.s.}{almost surely}
\newac{los}{LOS}{line-of-sight}
\newac{jsdm}{JSDM}{joint spatial division and multiplexing}
\newac{map}{MAP}{maximum a posteriori}
\newac{klt}{KLT}{Karhunen-Lo\`eve Transform}
\newac{lbe}{LBE}{link bargaining equilibrium}
\newac{se}{SE}{Stackelberg equilibrium}
\newac{uav}{UAV}{unmanned aerial vehicle}
\newac{nlos}{NLOS}{non-line-of-sight}
\newac{pdf}{PDF}{probability density function}
\newac{em}{EM}{expectation-maximization}
\newac{knn}{KNN}{$k$-nearest neighbor}
\newac{svd}{SVD}{singular value decomposition}
\newac{nmf}{NMF}{non-negative matrix factorization}
\newac{umf}{UMF}{unimodality-constrained matrix factorization}
\newac{rmse}{RMSE}{rooted mean squared error}
\newac{olos}{OLOS}{obstructed line-of-sight}
\newac{mmw}{mmW}{millimeter wave}
\newac{ber}{BER}{bit error rate}
\newac{rss}{RSS}{received signal strength}
\newac{lp}{LP}{linear program}
\newac{ufw}{U-FW}{unimodal Frank-Wolfe}
\newac{utf}{UTF}{unimodality-constrained tensor factorization}
\newac{fw}{FW}{Frank-Wolfe}
\newac{iot}{IoT}{Internet-of-Things}
\newac{mae}{MAE}{mean absolute error}
\newac{crb}{CRB}{Cram\'er-Rao bound}
\newac{aoa}{AoA}{angle of arrival}
\newac{wcl}{WCL}{weighted centroid localization}
\newac{thz}{THz}{Terahertz}
\newac{isac}{ISAC}{integrated sensing and communiation}
\newac{gnss}{GNSS}{global navigation satellite system}
\newac{3d}{3D}{three-dimensional}
\newac{2d}{2D}{two-dimensional}
\newac{pso}{PSO}{particle swarm optimization}
\newac{dp}{DP}{dynamic programming}
\newac{ao}{AO}{alternating optimization}
\newac{sca}{SCA}{successive convex approximation}

\makeatother

\usepackage{babel}
\addto\captionsamerican{\renewcommand{\lemmaname}{Lemma}}
\addto\captionsamerican{\renewcommand{\propositionname}{Proposition}}
\addto\captionsamerican{\renewcommand{\theoremname}{Theorem}}
\addto\captionsenglish{\renewcommand{\algorithmname}{Algorithm}}
\addto\captionsenglish{\renewcommand{\lemmaname}{Lemma}}
\addto\captionsenglish{\renewcommand{\propositionname}{Proposition}}
\addto\captionsenglish{\renewcommand{\theoremname}{Theorem}}
\providecommand{\lemmaname}{Lemma}
\providecommand{\propositionname}{Proposition}
\providecommand{\theoremname}{Theorem}

\begin{document}
\title{\selectlanguage{english}%
Active Search for Low-altitude UAV Sensing and Communication for Users
at Unknown Locations}

\maketitle
%
%



\ifdefined\SINGLECOLUMN
	\setkeys{Gin}{width=0.5\columnwidth}
	\newcommand{\figfontsize}{\footnotesize} 
	\newcommand{\labelfontsize}{0.7}
	\newcommand{\legendfontsize}{0.57}
	\newcommand{\ticklabelfontsize}{0.7}
	\newcommand{\numberfontsize}{0.45}
\else
	\setkeys{Gin}{width=1.0\columnwidth}
	\newcommand{\figfontsize}{\normalsize} 
	\newcommand{\labelfontsize}{0.7}
	\newcommand{\legendfontsize}{0.65}
	\newcommand{\ticklabelfontsize}{0.7}
	\newcommand{\numberfontsize}{0.5}
\fi
\newac{ode}{ODE}{ordinary differential equation}


\glsresetall
\begin{abstract}
This paper studies optimal \ac{uav} placement to ensure \ac{los}
communication and sensing for a cluster of ground users possibly in
deep shadow, while the \ac{uav} maintains backhaul connectivity with
a \ac{bs}. The key challenges include unknown user locations, uncertain
channel model parameters, and unavailable urban structure. Addressing
these challenges, this paper focuses on developing an efficient online
search strategy which jointly estimates channels, guides \ac{uav}
positioning, and optimizes resource allocation. Analytically exploiting
the geometric properties of the equipotential surface, this paper
develops an \ac{los} discovery trajectory on the equipotential surface
while the closed-form search directions are determined using perturbation
theory. Since the explicit expression of the equipotential surface
is not available, this paper proposes to locally construct a channel
model for each user in the \ac{los} regime utilizing polynomial regression
without depending on user locations or propagation distance. A class
of spiral trajectories to simultaneously construct the \ac{los} channels
and search on the equipotential surface is developed. An optimal radius
of the spiral and an optimal measurement pattern for channel gain
estimation are derived to minimize the mean squared error (MSE) of
the locally constructed channel. Numerical results on real 3D city
maps demonstrate that the proposed scheme achieves over $94\%$ of
the performance of a 3D exhaustive search scheme with just a $3$-kilometer
search.
\end{abstract}


\glsresetall

\section{Introduction}

\label{sec:Introduction}

In dense urban environments, short-wavelength signals like \ac{mmw},
\ac{thz}, and optical waves suffer from severe attenuation due to
signal blockage from buildings and other structures \cite{WanKonKonQiu:J18}.
In particular, sensing performance heavily relies on \ac{los} links
between targets and transceivers, which are usually disrupted in urban
environments \cite{CheLiuWanMas:J21,MenWuMaChe:J23}.

Deploying a low-altitude \ac{uav} is a promising solution for providing
\ac{los} links for users in deep shadow since the \ac{uav} operates
at a relatively high altitude with 3D mobility. Nevertheless, the
optimization of \ac{uav} deployment poses challenges due to the arbitrarily
complex terrain topology, which complicates modeling and predicting
\ac{los} conditions. Previous works bypassed this challenge by assuming
distance-dependent pure \ac{los} channel models \cite{LiuWanZhaChe:J19,NasTuaDuoPoo:J19},
or probabilistic channel models \cite{LiPanZhaHe:J20,YouZha:J20,MenGaoZhaYan:J22,XiaDonBaiWu:J20}.
These simplified channel models enable a coarse analysis but sacrifice
the performance by ignoring the actual \ac{los} condition. Some state-of-the-art
works formulated blockage-aware channel models \cite{YiZhuZhuXia:J22,KimSadLee:J23}
to assist \ac{uav} placement. For example, the authors in \cite{YiZhuZhuXia:J22}
approximated buildings as polyhedrons, and derived explicit expressions
to determine the \ac{los} status of a 3D position using an offline
city map. However, the performance of these methods highly relies
on the accuracy of the model of the environment.

In addition, most existing works \cite{LiuWanZhaChe:J19,NasTuaDuoPoo:J19,LiPanZhaHe:J20,YouZha:J20,MenGaoZhaYan:J22,XiaDonBaiWu:J20,YiZhuZhuXia:J22,KimSadLee:J23,ZheChe:C23}
assume perfect knowledge of user locations and channel models. Nevertheless,
in practical scenarios, accurate user locations and precise channel
models are typically unavailable due to the location privacy of users
and the uncertain \ac{nlos} conditions. Some recent works considered
first estimating user locations and channel parameters before optimizing
\ac{uav} service positions \cite{EsrGanGes:J19,KriHanCab:C19,EsrGanGes:J21}.
However, localization itself under possibly \ac{nlos} conditions
is already a challenging problem and the performance of \ac{uav}
placement highly depends on the localization accuracy of the users.
Assuming the availability of 3D maps, the work \cite{EsrGanGes:J21}
employed the fisher information metric to design an online trajectory
for joint path-loss parameter estimation and user localization, but
the global optimality and complexity remain unclear. With known user
locations and maps, the work \cite{EsrGanGes:J19} utilized dynamic
programming to optimize \ac{uav} trajectory for channel parameters
learning and employed map compression to estimate the \ac{los} probability
for \ac{uav} positions. However, the placement performance relies
on channel learning and \ac{los} probability estimation. Without
the prior knowledge of channels and user locations, a model-free approach
in \cite{KriHanCab:C19} relied on \ac{sinr} measurements and a 3D
map to account for blockage and scatters, and the authors applied
Q-learning to optimize \ac{uav} placement. However, the action space
is restricted to four orthogonal horizontal directions on a 2D plane,
and there is little theoretical guarantee on the performance.

This paper investigates the \ac{uav} trajectory design to provide
communication and sensing services for a cluster of ground users at
unknown locations, while the \ac{uav} maintains backhaul connectivity
with a remote \ac{bs}. The goal is to provide an optimization framework
for an efficient online search of the best \ac{uav} position for
a broad class of communication problems and sensing problems. The
challenges arise from the need to operate without precise information,
where the user locations and the channel models are not available
at the \ac{uav}, and the \ac{uav} has to make search decisions and
adjust the search trajectory on-the-fly based on the online measurements.

In this work, we develop three techniques to resolve the above challenges.
First, we propose to search on an {\em equipotential surface} using
trajectories developed from {\em perturbation theory}. The equipotential
surface is the region of UAV positions where the sensing and communication
performance for the users is balanced with capacity of the UAV-BS
backhaul link. The status whether the UAV is on or off the equipotential
surface can be quantified by measuring the \ac{snr}, without requiring
knowledge of user locations or channel parameters. While an explicit
expression of the equipotential surface is not available, we employ
perturbation theory to develop a search strategy so that the search
trajectory remains on the equipotential surface. Second, we propose
to locally construct a channel map for each user within the \ac{los}
regime using local polynomial regression. This approach allows for
a relatively simple construction of nonparametric channel models that
are capable of capturing actual signal attenuations due to blockage
from the possibly unknown propagation environment. Third, we develop
our search strategy on the equipotential surface exploiting two universal
properties: {\em upward invariance} and {\em colinear invariance}
of \ac{los} regions over almost all terrain structure. Prior studies
have shown that, in certain cases, this strategy can achieve an $\epsilon$-optimal
solution globally in 3D space with a trajectory length linear of the
search radius \cite{ZheChe:J23,ZheCheLuo:C23}. Yet, the prior work
\cite{ZheChe:J23,ZheCheLuo:C23} assumed knowledge of user locations
and channel model parameters.

The contributions of this paper are summarized as follows:
\begin{itemize}
\item We analytically show that the equipotential surface is a sphere for
a class of sensing and communication problems, where the sphere parameters
depend on the user distribution and the power budget.
\item We develop a class of spiral trajectories to simultaneously construct
the local LOS channels and search on the equipotential surface. An
optimal radius of the spiral and an optimal measurement pattern for
channel gain estimation are derived to minimize the \ac{mse} of the
locally constructed channel.
\item We demonstrate that the normalized channel gain construction error
is on the order of $10^{-2}$ without knowing the user location or
the propagation distance. With a $3$-kilometer search, the proposed
scheme can achieve over $94\%$ performance of that from a 3D exhaustive
search for a UAV-assisted multiuser sensing and communication problem
over a dense urban area.
\end{itemize}

The remaining part of the paper is organized as follows: Section II
introduces the system model and formulates a UAV-assisted sensing
and communication problem. Section III discusses the geometric properties
of the equipotential surface and local channel map construction. Section
IV outlines the LOS discovery trajectory and proposes a superposed
trajectory for optimal LOS position search. Section V presents numerical
results and comparisons, while Section VI concludes the paper.

\textit{Notation}: Vectors and matrices are denoted by bold $\mathbf{x}$
and bold capital $\mathbf{X}$, respectively. $\mathcal{M}_{m,n}$
denotes all $m$-by-$n$ matrices with $\mathcal{M}_{n}$ for square
matrices. Matrix entry, column vector, and row vector are represented
as $[\mathbf{X}]_{(i,j)}$, $[\mathbf{X}]_{(:,j)}$, and $[\mathbf{X}]_{(i,:)}$,
respectively. Matrix trace and diagonal are $\text{tr}\{\mathbf{X}\}$
and $\text{diag}\{\mathbf{X}\}$. Expectation and variance are $\mathbb{E}\{\cdot\}$
and $\mathbb{V}\{\cdot\}$. The gradient of $f(\mathbf{x})$ is $\ensuremath{\nabla f(\mathbf{x})}$.
The cross product is indicated by $\times$. The time derivative of
$x(t)$ is defined as $\dot{x}=\text{d}x(t)/\text{d}t$. $C$ represents
a constant. The inequality $\mathbf{p}\geq0$ indicates that all the
entries in $\mathbf{p}$ are no less than $0$.

\section{System Model}

\label{subsec:System-model}
In this section, we first establish an environment model with certain
\ac{los} properties. Then, a channel model for \ac{los} propagation
condition and \ac{nlos} propagation condition is defined. Finally,
a general \ac{los}-guaranteed \ac{uav}-assisted sensing and communication
problem is formulated with two specific problems followed.

\subsection{Environment Model}

\label{subsec:Environment-and-Propagation-model}

Consider that a \ac{uav} serves one \ac{bs} located at $\mathbf{u}_{0}$
and a cluster of users without knowing their locations as shown in
Fig.~\ref{fig:system-illustration}. We focus on the case where the
users are clustered in an unknown neighborhood in a dense urban area.
The sets of users are denoted as $\mathcal{K}\triangleq\{1,2,\dots,K\}$.

Let $\mathcal{X}=\{\mathbf{x}\in\mathbb{R}^{3}:x_{3}\geq H_{\text{min}}\}$
be the feasible region of \ac{uav} positions where $H_{\text{min}}$
is the minimum flight height. While signals of ground users are likely
blocked by buildings, we denote $\mathcal{D}_{k}\text{\ensuremath{\subseteq}}\mathcal{X}$
as the region of \ac{uav} positions such that there is an \ac{los}
link between the \ac{uav} at position $\mathbf{x}\in\mathcal{D}_{k}$
and user $k$ at an unknown position $\mathbf{u}_{k}$.

The \ac{los} region $\mathcal{D}_{k}$ can be {\em arbitrary} except
that $\mathcal{D}_{k}$ is assumed to have the following properties:
For any $\mathbf{x}\in\mathcal{D}_{k}$,
\begin{enumerate}
\item Upward invariant: any \ac{uav} position $\mathbf{x}'$ perpendicularly
above $\mathbf{x}$ also belongs to $\mathcal{D}_{k}$, \emph{i.e.},
$\mathbf{x}'\in\mathcal{D}_{k}$;
\item Colinear invariant: any \ac{uav} position $\mathbf{x}'$ that satisfies
$\mathbf{x}'-\mathbf{u}_{k}=\rho(\mathbf{x}-\mathbf{u}_{k})$ for
some $\rho>1$ also belongs to $\mathcal{D}_{k}$.
\end{enumerate}

Similarly, one can define $\mathcal{D}_{0}\subseteq\mathcal{X}$ as
the \ac{los} region of \ac{uav} positions to the \ac{bs}. To summarize,
the upward invariant and colinear invariant properties imply that
if there is an \ac{los} link between the \ac{uav} and a user, such
an \ac{los} condition will remain if the \ac{uav} increases its
altitude or moves away from the user without changing the elevation
and azimuth angles. The widely adopted probabilistic \ac{los} model
in the \ac{uav} literature \cite{CheMozSaaYin:J17,IreSebHal:J18,CheHua:J22,SinAgrSinBan:J22}
is a special case that satisfies these properties in a statistical
sense.

Define the {\em full-LOS} region $\tilde{\mathcal{D}}=\bigcap\mathcal{D}_{k}$
as the set of \ac{uav} positions where there are \ac{los} links
to the \ac{bs} and all the users. Since the full-LOS region is an
intersection of $\mathcal{D}_{k}$, the upward invariant property
automatically holds, {\em i.e.}, for any full-LOS position $\mathbf{x}\in\tilde{\mathcal{D}}$,
any position $\mathbf{x}'$ perpendicularly above $\mathbf{x}$ is
also a full-LOS position which satisfies $\mathbf{x}'\in\tilde{\mathcal{D}}$.
Note that the colinear invariant property does not hold for $\tilde{\mathcal{D}}$.

\subsection{Channel Model}

\label{subsec:Channel-Model}

Based on the LOS region $\mathcal{D}_{k}$, the channel gain from
the \ac{uav} at position $\mathbf{x}$ to node $k$ at the unknown
position $\mathbf{u}_{k}$ is modeled as
\begin{equation}
\bar{g}_{k}(\mathbf{x})=\begin{cases}
\begin{aligned} & g_{k}(\mathbf{x}) &  & \text{if }\mathbf{x}\in\mathcal{D}_{k}\quad\text{(1a)}\\
 & g_{k}(\mathbf{x})+\phi(\mathbf{x}) &  & \text{otherwise}\quad\,\text{(1b)}
\end{aligned}
\end{cases}\label{eq:propagation-model}
\end{equation}
where $g_{k}(\mathbf{x})$ is the deterministic channel gain under
\ac{los}, and $\phi(\mathbf{x})$ is a random variable to capture
the power penalty due to the shadowing in \ac{nlos} condition \cite{CheGes:J19}.

\setcounter{equation}{1}

The global model for the deterministic channel gain $g_{k}(\mathbf{x})$
is {\em unknown} to the system, except that $g_{k}(\mathbf{x})$
is assumed to be Lipschitz continuous satisfying,
\begin{align}
g_{k}(\mathbf{x}) & \leq g_{k}(\mathbf{x}_{0})+\nabla g_{k}(\mathbf{x}_{0})^{\text{T}}(\mathbf{x}-\mathbf{x}_{0})+\frac{L_{g}}{2}d^{2}(\mathbf{x},\mathbf{x}_{0})\label{eq:Lipschitz-g-upper}
\end{align}
and
\begin{align}
g_{k}(\mathbf{x}) & \geq g_{k}(\mathbf{x}_{0})+\nabla g_{k}(\mathbf{x}_{0})^{\text{T}}(\mathbf{x}-\mathbf{x}_{0})-\frac{L_{g}}{2}d^{2}(\mathbf{x},\mathbf{x}_{0})\label{eq:Lipschitz-g-lower}
\end{align}
for all $\mathbf{x}\in\mathcal{X}$ where $L_{g}$ is a finite constant,\footnote{The free-space propagation model, {\em i.e.}, $g_{k}(\mathbf{x})=b_{0}+10a_{0}\log_{10}(d(\mathbf{x},\mathbf{u}_{k}))$
where $b_{0}$ and $a_{0}$ are channel parameters, satisfies this
condition with $L_{g}=8.7\times\text{10}^{-4}$ under $a_{0}=-2$
and $H_{\text{min}}=100$ meters.} and $d(\mathbf{x},\mathbf{x}_{0})\triangleq\|\mathbf{x}-\mathbf{x}_{0}\|_{2}$.

While the global model $g_{k}({\bf x})$ or $\bar{g}_{k}({\bf x})$
is not available, the \ac{uav} may measure the channel gain $\bar{g}_{k}(\mathbf{x})$
when it explores location $\mathbf{x}$. The measurement model is
given by 
\begin{equation}
y=\bar{g}_{k}(\mathbf{x})+\xi\label{eq:measurement_with_noise}
\end{equation}
 where $\xi\sim\text{\ensuremath{\mathcal{N}}}(0,\sigma^{2})$ models
the small-scale fading and the uncertainty due to the antenna gain
which may not be omidirectional.

In addition, the complete geometry of $\mathcal{D}_{k}$ is also unknown,
except for a local area that has been explored by the \ac{uav} along
its trajectory up to time $t$. Specifically, denote $\mathbf{x}(t)$
as the \ac{uav} position at time $t$. We assume that the value of
\ac{los} indicator function $\mathbb{I}\{\mathbf{x}(t)\in\mathcal{D}_{k}\}$
can be perfectly determined based on the measurements $\bar{g}_{k}(\mathbf{x}(\tau))$
for $0\leq\tau\leq t$. For a practical implementation, $\mathbb{I}\{\mathbf{x}(t)\in\mathcal{D}_{k}\}$
can be computed by statistical learning and hypothesis testing \cite{CheGes:J19,ZheChe:C23}.
\begin{figure}
\begin{centering}
\includegraphics[width=0.95\columnwidth]{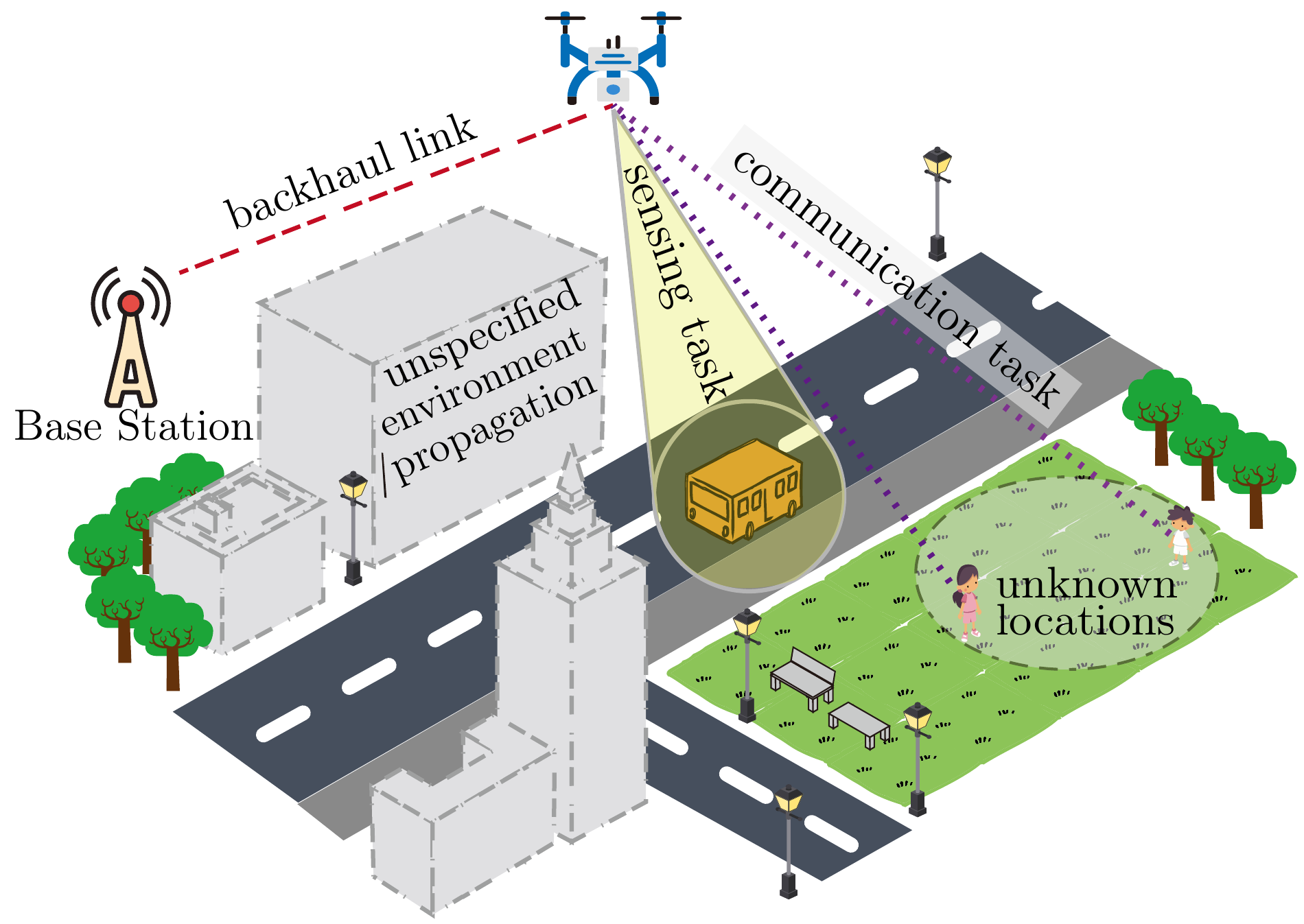}
\par\end{centering}
\caption{\label{fig:system-illustration} Illustration of a system where a
\ac{uav} provides sensing and/or communication services for a cluster
of sensing targets and/or communication users without knowing their
locations, channel models, and the city topology while establishing
an \ac{los} backhaul link to a \ac{bs}.}
\end{figure}

\subsection{UAV-assisted Sensing and Communication with Backhaul}

\label{subsec:A-UAV-assisted-sensing-and-communication-problem}

A common joint \ac{uav} position optimization and resource allocation
problem needs to balance the performance of serving the users and
the capacity of the backhaul link that the \ac{uav} connects to the
\ac{bs}. In addition, most sensing tasks require an \ac{los} condition.
Likewise, for a communication task, the \ac{los} condition can substantially
enhance the communication performance. These requirements lead to
a general max-min problem for the {\em\ac{los}-guaranteed} \ac{uav}-assisted
sensing and communication as follows.

Denote $\bm{g}_{\text{u}}(\mathbf{x})$ as a vector that collects
the \ac{los}  channel gains for all the users, {\em i.e.}, $\bm{g}_{\text{u}}(\mathbf{x})=[g_{1}(\mathbf{x}),\dots,g_{K}(\mathbf{x})]^{\text{T}}$.
Denote $\mathbf{p}$ as the corresponding resource allocation given
$\bm{g}_{\text{u}}(\mathbf{x})$. We have 
\begin{equation}
\begin{aligned}\mathscr{P}:\quad\mathop{\mbox{maximize}}\limits _{\mathbf{x},\mathbf{p}} & \quad\min\{f_{0}(g_{0}(\mathbf{x})),F_{\text{u}}(\bm{g}_{\text{u}}(\mathbf{x}),\mathbf{p})\}\\
\mathop{\mbox{subject to}} & \quad\mathbf{x}\in\text{\ensuremath{\tilde{\mathcal{D}}}},\\
 & \quad H_{n}(\bm{g}_{\text{u}}(\mathbf{x}),\mathbf{p})\leq0,n=1,2,\dots,N
\end{aligned}
\label{general_formula}
\end{equation}
where $f_{0}(g_{0}(\mathbf{x}))$ represents the objective of the
BS-UAV link under the \ac{los} condition, $F_{\text{u}}(\bm{g}_{\text{u}}(\mathbf{x}),\mathbf{p})$
represents the objective of the UAV-user links under \ac{los}, and
$H_{n}(\bm{g}_{\text{u}}(\mathbf{x}),\mathbf{p})\leq0$ for $n=1,2,\dots,N$
are the corresponding constraints for the resource allocation. Problem~$\mathscr{P}$
is non-convex due to the \ac{los} constraint on \ac{uav} positions
and the fact that the blockage can have an arbitrary shape.

Such a general formulation captures many typical applications for
communication and sensing, with two examples illustrated as follows.

\subsubsection{Balancing problem\label{subsec:Balancing-problem}}

Consider to deploy a \ac{uav} to offer sensing and/or communication
services for a cluster of sensing targets and/or communication users
while establishing a \ac{los} relay link with a remote \ac{bs}.
Specify $p_{k}$ as the power allocation from the \ac{uav} to user
$k$. Denote $f_{k}(g_{k}(\mathbf{x}),p_{k})$ as the objective function
for user $k$, and $N_{0}$ as the noise power of the propagation
channel. For sake of brevity, let $N_{0}=1$ in this paper. For a
sensing task involving estimation, $f_{k}(g_{k}(\mathbf{x}),p_{k})$
can be specified as weighted \ac{snr}, {\em i.e.}, $f_{k}(g_{k}(\mathbf{x}),p_{k})=\mu_{\text{s}}p_{k}g_{k}(\mathbf{x})$
where $\mu_{\text{s}}$ is a weight of the sensing task \cite{JohVenGroLop:J22,ZviYon:J10}.
For a communication task, $f_{k}(g_{k}(\mathbf{x}),p_{k})$ can be
specified as channel capacity, {\em i.e.}, $f_{k}(g_{k}(\mathbf{x}),p_{k})=\mu_{\text{c}}\log_{2}(1+p_{k}g_{k}(\mathbf{x}))$
where $\mu_{\text{c}}$ is a weight of the communication task. In
addition, the objective function of the BS-UAV link is given by the
capacity function, {\em i.e.}, $f_{0}(g_{0}(\mathbf{x}))=\log_{2}(1+P_{0}g_{0}(\mathbf{x}))$
where $P_{0}$ is the transmit power of the \ac{bs}. The balancing
problem aims at maximizing the worst link performance of the BS-UAV
link and the UAV-user links. Hence, the overall objective function
is specified as $\min\{f_{0}(g_{0}(\mathbf{x})),\min_{k\in\mathcal{K}}\{f_{k}(g_{k}(\mathbf{x}),p_{k})\}\}$.
The balancing problem jointly optimizes the \ac{uav} position $\mathbf{x}$
and the power allocation $\mathbf{p}$ as follows
\begin{equation}
\begin{aligned}\mathop{\mbox{maximize}}\limits _{\mathbf{x},\mathbf{p}\geq0} & \quad\min\{f_{0}(g_{0}(\mathbf{x})),\min_{k\in\mathcal{K}}\{f_{k}(g_{k}(\mathbf{x}),p_{k}\text{)}\}\}\\
\mathop{\mbox{subject to}} & \quad\mathbf{x}\in\text{\ensuremath{\tilde{\mathcal{D}}}},\\
 & \quad\text{\ensuremath{\sum_{k\in\mathcal{K}}}}p_{k}\leq P_{\text{T}}
\end{aligned}
\label{general_formula-balancing-problem}
\end{equation}
where $P_{\text{T}}$ is the total transmit power of the \ac{uav}.

\subsubsection{Sum-rate problem\label{subsec:Sum-rate-problem}}

Consider that a \ac{uav} relays signal from a \ac{bs} to $K$ ground
users with other assumptions the same as that in the balancing problem.
The sum-rate problem aims at maximizing the sum capacity of the relay
channels. Thus, the objective function is $\min\{f_{0}(g_{0}(\mathbf{x})),\sum_{k\in\mathcal{K}}f_{k}(g_{k}(\mathbf{x}),p_{k})\}$
where $f_{0}(g_{0}(\mathbf{x}))$ and $f_{k}(g_{k}(\mathbf{x}),p_{k})$
are defined in Section~\ref{subsec:Balancing-problem}, and the problem
is formulated as
\begin{equation}
\begin{aligned}\mathop{\mbox{maximize}}\limits _{\mathbf{x},\mathbf{p}\text{\ensuremath{\geq}}0} & \quad\min\{f_{0}(g_{0}(\mathbf{x})),{\textstyle \sum_{k\in\mathcal{K}}}f_{k}(g_{k}(\mathbf{x}),p_{k})\}\\
\mathop{\mbox{subject to}} & \quad\mathbf{x}\in\text{\ensuremath{\tilde{\mathcal{D}}}},\\
 & \quad\sum_{k\in\mathcal{K}}p_{k}\leq P_{\text{T}}.
\end{aligned}
\label{general_formula-sum-rate}
\end{equation}

While the example formulation~(\ref{general_formula-balancing-problem})
is non-convex in the power allocation variable $\mathbf{p}$ and the
problem~(\ref{general_formula-sum-rate}) is convex in $\mathbf{p}$,
both cases can be handled in a same way in our proposed algorithm
framework. In addition, as the channels and the LOS regions $\mathcal{D}_{k}$
are not available before exploring near location $\mathbf{x}$, the
\ac{uav} needs to design an online trajectory to explore the \ac{los}
opportunity, measure the channel quality, and optimize for the system
performance.

\subsection{Suboptimal Solution on the Equipotential Surface}

As the full \ac{los} region $\tilde{\mathcal{D}}$ can have an arbitrary
shape and is initially unknown before the exploration, finding the
globally optimal solution to $\mathscr{P}$ generally requires an
online exhaustive search in 3D, which is prohibitive due to the limited
flight time of \ac{uav}. Thus, we compromise for a suboptimal solution
on the \emph{equipotential surface}.

The equipotential surface $\mathcal{S}$ is defined as a region where
the objective of the BS-UAV link and the objective of the UAV-user
links under the optimized resource allocation are equal, assuming
all the links were in \ac{los}. Specifically, define $\mathbf{p}^{*}(\mathbf{x})$
as the optimal solution to $\mathscr{P}$ given a fixed location $\mathbf{x}$
by ignoring the \ac{los} constraint $\mathbf{x}\in\text{\ensuremath{\tilde{\mathcal{D}}}}$.
Denote $\bm{g}(\mathbf{x})=[g_{0}(\mathbf{x}),\bm{g}_{\text{u}}(\mathbf{x})^{\text{T}}]^{\text{T}}$
as a vector that collects the \ac{los} channel gains from the \ac{bs}
and the users. Defining a balance function 
\begin{equation}
F(\bm{g}(\mathbf{x}))\text{\ensuremath{\triangleq}}f_{0}(g_{0}(\mathbf{x}))-F_{\text{u}}(\bm{g}_{\text{u}}(\mathbf{x}),\mathbf{p}^{*}(\mathbf{x}))\label{eq:F-gx}
\end{equation}
the equipotential surface is defined as
\begin{equation}
\mathcal{S}=\left\{ \mathbf{x}\in\mathcal{X}:F(\bm{g}(\mathbf{x}))=0\right\} .\label{eq:equipotential_surface}
\end{equation}

Recent studies \cite{ZheCheLuo:C23,ZheChe:C23,ZheChe:J23} discover
that searching on the equipotential surface $\mathcal{S}$ has significant
promise in identifying the globally optimal \ac{uav} position in
3D space. It has been shown that for the case with a single user and
a \ac{bs}, a solution can be found with only an $O(\epsilon)$ performance
gap to the globally optimal solution by searching {\em only} on the
equipotential surface, for a search distance of $O(1/\epsilon)$,
where the equipotential surface degenerates to a middle perpendicular
plane under the condition that $P_{0}=P_{\text{T}}$ \cite{ZheChe:C23,ZheChe:J23}.
It was also numerically demonstrated in \cite{ZheCheLuo:C23}, where
the user locations are known, that searching on the equipotential
surface in a multi-user case for a sum-rate maximization objective
attains over $96\%$ of the performance of that from an exhaustive
search in the entire 3D space.

\subsection{Superposed Trajectory for Unknown User Locations}

Since the user locations $\mathbf{u}_{k}$ are unknown, the analytical
form of the channels $\bm{g}(\mathbf{x})$ is \emph{not} available,
and hence, \emph{no} analytical form of the equipotential surface
$\mathcal{S}$ is available to the system. As a result, while the
\ac{uav} aims at searching on $\mathcal{S}$, it also needs to simultaneously
estimate $\bm{g}(\mathbf{x})$ and construct $\mathcal{S}$. A classical
approach may first estimate a parametric form of $\bm{g}(\mathbf{x})$
and then construct a global analytical model for $\mathcal{S}$. However,
it is known that a joint user localization and propagation parameter
estimation for $\bm{g}(\mathbf{x})$ require measurements across a
large area, and such global construction is very challenging and inaccurate.

To tackle this challenge, we develop a superposed trajectory for the
\ac{uav} exploration as $\mathbf{x}(t)=\text{\textbf{x}}_{\text{s}}(t)+\text{\textbf{x}}_{\text{r}}(t)$
where $\text{\textbf{x}}_{\text{s}}(t)$ is the online search trajectory
to find a suboptimal solution to $\mathscr{P}$ on $\mathcal{S}$
based on the local information of $\bm{g}(\mathbf{x})$ in the neighborhood
along the search, and $\text{\textbf{x}}_{\text{r}}(t)$ provides
small deviation from $\text{\textbf{x}}_{\text{s}}(t)$ to simultaneously
collect measurements for the \emph{local} reconstruction of $\bm{g}(\mathbf{x})$
and $\mathcal{S}$.

\section{Trajectory Design for Tracking on the Equipotential Surface}

\label{sec:Trajectory-Design-for-Equipotential-Surface}

In this section, we first explore the geographic properties of $\mathcal{S}$,
and subsequently, demonstrate the feasibility of approximately constructing
the equipotential surface without knowing the channel models and the
user locations. Then, the local construction of the propagation model
is studied with theoretical results on the optimal measurement pattern.
A class of spiral trajectories for locally constructing the equipotential
surface is proposed at the end of this section.

\subsection{Property of the Equipotential Surface}

\label{subsec:Property-of-the-Equipotential-Surface}

The equipotential surface does not exist when $\mathcal{S}=\text{\ensuremath{\varnothing}}$.
This corresponds to a superior channel for the \ac{bs} or for the
users, resulting in a trivial solution where the \ac{uav} should
hover above either the \ac{bs} or above the cluster of the users.
We focus on the scenario where $\mathcal{S}$ does exist.

\subsubsection{Existence Condition}

\label{subsec:Existence-Condition}

The existence of the equipotential surface can be easily checked by
evaluating the objective values at two special locations $\mathbf{x}_{0}^{\text{m}}=\mathbf{u}_{0}+[0,0,H_{\text{min}}]^{\text{T}}$
and $\mathbf{x}_{\text{u}}^{\text{m}}=\sum_{k\in\mathcal{K}}\mathbf{u}_{k}/K+[0,0,H_{\text{min}}]^{\text{T}}$.
A general existence condition yields $F(\bm{g}(\mathbf{x}_{0}^{\text{m}}))F(\bm{g}(\mathbf{x}_{\text{\text{u}}}^{\text{m}}))\leq0$.
This is because both the channel $\bm{g}(\mathbf{x})$ and $F(\bm{g}(\mathbf{x}))$
are continuous, and hence, there exists a path from $\mathbf{x}_{\text{u}}^{\text{m}}$
to $\mathbf{x}_{0}^{\text{m}}$ that reaches $F(\bm{g}(\mathbf{x}))=0$.

For a specific problem, such as the balancing problem in (\ref{general_formula-balancing-problem}),
the existence condition depends on the problem parameters, such as
the power budget $P_{\text{T}}$.
\begin{prop}[Existence condition in a specified balancing problem]
\label{prop:existence-condition-special-balancing-problem}For the
balancing problem defined in (\ref{general_formula-balancing-problem})
with $f_{k}(g_{k}(\mathbf{x}),p_{k})=\log_{2}(1+p_{k}g_{k}(\mathbf{x}))$,
a sufficient condition to the existence of the equipotential surface
is given by
\begin{align}
 & \left(\frac{P_{0}}{P_{{\rm T}}}-\frac{1/g_{0}(\mathbf{x}_{0}^{\text{{\rm m}}})}{\sum_{k\in\mathcal{K}}(1/g_{k}(\mathbf{x}_{0}^{{\rm m}}))}\right)\left(\frac{P_{0}}{P_{{\rm T}}}-\frac{1/g_{0}({\rm \mathbf{x}_{\text{{\rm u}}}^{{\rm m}}})}{\sum_{k\in\mathcal{K}}(1/g_{k}({\rm \mathbf{x}_{{\rm u}}^{{\rm m}}}))}\right)\leq0.\label{eq:sufficient_condition_existence-balancing-problem}
\end{align}
\end{prop}
\begin{proof}
See Appendix~\ref{sec:Proof-of-Proposition-existence-condition}.
\end{proof}
Proposition~\ref{prop:existence-condition-special-balancing-problem}
provides an explicit condition to the existence of the equipotential
surface in a typical balancing problem for maximizing the worst relay
channel capacity. In particular, when $\mathcal{K}=\{1\}$, $1/(g_{0}(\mathbf{x})\sum_{k\in\mathcal{K}}(1/g_{k}(\mathbf{x})))$
is simplified as $g_{1}(\mathbf{x})/g_{0}(\mathbf{x})$, and thus
the sufficient condition in (\ref{eq:sufficient_condition_existence-balancing-problem})
becomes $(P_{0}/P_{\text{T}}-g_{1}(\mathbf{x}_{0}^{\text{m}})/g_{0}(\mathbf{x}_{0}^{\text{m}}))(P_{0}/P_{\text{T}}-g_{1}(\mathbf{x}_{\text{u}}^{\text{m}})/g_{0}(\mathbf{x}_{\text{u}}^{\text{m}}))\leq0$.
Consequently, there exists a point $\mathbf{x}$ satisfying $P_{0}/P_{\text{T}}=g_{1}(\mathbf{x})/g_{0}(\mathbf{x})$
implying that a smaller channel gain necessitates a corresponding
increment in the allocated power budget.

\subsubsection{Geometric Shape}

\label{subsec:Geometric-Characterization}

While the geometric characteristics of the equipotential surface $\mathcal{S}$
are highly related to the specific applications, the exploitation
of geometric properties of $\mathcal{S}$ in a typical balancing problem
provides some insights for \ac{uav} trajectory design on $\mathcal{S}$.
\begin{prop}[A spherical equipotential surface]
\label{prop:shape_of_equipotential_surface}Suppose that the channel
gain satisfies $g_{k}(\mathbf{x})=b_{0}-10\log_{10}(d^{2}(\mathbf{x},\mathbf{u}_{k}))$.
For the balancing problem defined in (\ref{general_formula-balancing-problem})
with $f_{k}(g_{k}(\mathbf{x}),p_{k})=\log_{2}(1+p_{k}g_{k}(\mathbf{x}))$,
the equipotential surface is a sphere centered at $\mathbf{o}=(P_{0}\sum_{k\in\mathcal{K}}\mathbf{u}_{k}-P_{\text{T}}\mathbf{u}_{0})/(KP_{0}-P_{\text{T}})$
with radius $R$ satisfying 
\begin{align}
 & R^{2}=\frac{P_{\text{T}}\|\mathbf{u}_{0}\|_{2}^{2}-P_{0}\sum_{k\in\mathcal{K}}\|\mathbf{u}_{k}\|_{2}^{2}}{KP_{0}-P_{\text{T}}}+\|\mathbf{o}\|_{2}^{2}\label{eq:radius_2}
\end{align}
\end{prop}
\begin{proof}
See Appendix~\ref{sec:Proof-of-Proposition-sphere}.
\end{proof}
Knowing that the equipotential surface $\mathcal{S}$ is a sphere,
one may estimate the center and radius of the sphere, and hence, one
can easily design a trajectory exploring \ac{los} opportunity on
$\mathcal{S}$. In general, when $\mathcal{S}$ is not guaranteed
to be a sphere, it is also inspired from Proposition~\ref{prop:shape_of_equipotential_surface}
that $\mathcal{S}$ may be locally approximated by a patch from a
sphere, leading to some simplified design of local trajectories. In
particular, if $K=1$ and $P_{\text{T}}=P_{0}$, the equipotential
surface becomes a middle-perpendicular plane between the \ac{bs}
and the user satisfying $d(\mathbf{x},\mathbf{u}_{0})=d(\mathbf{x},\mathbf{u}_{1})$
for $\forall\mathbf{x}\in\mathcal{S}$ \cite{ZheChe:J23}.

\subsection{Trajectory towards the Equipotential Surface}

\label{subsec:A-Local-Approximation-of-the-Equipotential-Surface}

Recall that the channel gain $\bm{g}(\mathbf{x})$ is not available
before the UAV explores location $\mathbf{x}$, and moreover, the
analytical form of $F(\bm{g}(\mathbf{x}))$ in (\ref{eq:F-gx}) is
not available. Here, we develop an iterative exploration strategy
to move towards the equipotential surface $\mathcal{S}$ that satisfies
$F(\bm{g}(\mathbf{x}))=0$.

We adopt a gradient-type search. Whenever the \ac{uav} locates off
the equipotential surface $F(\bm{g}(\mathbf{x}))=\delta\neq0$ at
$\mathbf{x}=\mathbf{c}_{0}$, it moves to the direction that steepest
decreases (or increases) $F(\bm{g}(\mathbf{x}))$. A linear approximation
of $F(\bm{g}(\mathbf{x}))$ at $\mathbf{x}=\mathbf{c}_{0}$ yields
\begin{equation}
F(\bm{g}(\mathbf{x}))\approx F(\bm{g}(\mathbf{c}_{0}))+\nabla F(\bm{g}(\mathbf{c}_{0}))^{\text{T}}\mathbf{G}(\mathbf{c}_{0})(\mathbf{x}-\mathbf{c}_{0})\label{eq:linear-approximation-of-Fgx}
\end{equation}
where $\nabla F(\bm{g}(\mathbf{c}_{0}))^{\text{T}}\mathbf{G}(\mathbf{c}_{0})$
represents the gradient of $F(\bm{g}(\mathbf{x}))$ at $\mathbf{x}=\mathbf{c}_{0}$
and $\mathbf{G}(\mathbf{x})=[\nabla g_{0}(\mathbf{x}),\nabla g_{1}(\mathbf{x}),\dots,\nabla g_{K}(\mathbf{x})]^{\text{T}}$
is the matrix that collects the gradients of the channel gains $\bm{g}(\mathbf{x})$.

Setting the above linear approximation~(\ref{eq:linear-approximation-of-Fgx})
to $0$ and noticing that $F(\bm{g}(\mathbf{c}_{0}))=\delta$, a nearest
solution $\mathbf{x}$ from $\mathbf{c}_{0}$ can be found by
\begin{equation}
\begin{aligned}\mathop{\mbox{minimize}}\limits _{\mathbf{x}} & \quad\|\mathbf{x}-\mathbf{c}_{0}\|_{2}^{2}\\
\mathop{\mbox{subject to}} & \quad\nabla F(\bm{g}(\mathbf{c}_{0}))^{\text{T}}\mathbf{G}(\mathbf{c}_{0})(\mathbf{x}-\mathbf{c}_{0})=-F(\bm{g}(\mathbf{c}_{0}))
\end{aligned}
\label{eq:distance-minimization-problem}
\end{equation}
where the closed-form solution is obtained using the Lagrangian multiplier
method as $\hat{\mathbf{x}}=\mathbf{c}_{0}+\mathcal{V}(\mathbf{c}_{0})$,
where
\begin{equation}
\mathcal{V}(\mathbf{c}_{0})\triangleq-\frac{\mathbf{G}(\mathbf{c}_{0})^{\text{T}}\nabla F(\bm{g}(\mathbf{c}_{0}))F(\bm{g}(\mathbf{c}_{0}))}{\|\mathbf{G}(\mathbf{c}_{0})^{\text{T}}\nabla F(\bm{g}(\mathbf{c}_{0}))\|_{2}^{2}}.\label{eq:solution-to-local-approximation-equipotential-point}
\end{equation}
The solution $\hat{\mathbf{x}}$ provides an estimated location on
$\mathcal{S}$. Consequently, the \emph{one-step exploration range}
for the \ac{uav} to approximately reach $\mathcal{S}$ is given by
$r_{0}\triangleq\|\hat{\mathbf{x}}-\mathbf{c}_{0}\|_{2}$.

As a result, a search trajectory for tracking $\mathcal{S}$ can be
designed as $\mathbf{x}_{\text{s}}(t_{n+1})=\mathbf{x}_{\text{s}}(t_{n})+\mathcal{V}(\mathbf{x}_{\text{s}}(t_{n}))$.
For mathematical convenience, the search trajectory $\mathbf{x}_{\text{s}}(t)$
that is described by a continuous-time \ac{ode} is given by 
\begin{equation}
\dot{\mathbf{x}}_{\text{s}}=\text{d}\mathbf{x}_{\text{s}}(t)/\text{d}t=\mu_{\text{v}}\mathcal{V}(\mathbf{x}_{\text{s}}(t))\label{eq:equipotential-surface-tracking}
\end{equation}
where $\mu_{\text{v}}>0$ is a step size to control the speed of the
trajectory.

It is observed that computing the search direction $\mathcal{V}(\mathbf{x})$
in (\ref{eq:solution-to-local-approximation-equipotential-point})
not only requires the channel gains $\bm{g}(\mathbf{x})$, but also
its gradient $\mathbf{G}(\mathbf{x})$. In the rest of this section,
we develop methods and trajectory $\mathbf{x}_{\text{r}}(t)$ to locally
construct $\bm{g}(\mathbf{x})$ and its gradient $\mathbf{G}(\mathbf{x})$
at the neighborhood of $\mathbf{x}_{\text{s}}(t)$.

\subsection{Construction of a Local Channel Map}

\label{subsec:Construction-of-a-Local-Channel-Model}

It is well-known that directly estimating a global propagation model
$g_{k}(\mathbf{x})$ is difficult, as it leads to nonlinear regression.
Instead, for each node $k$, we adopt a linear model $\hat{g}(\mathbf{x})$
to locally approximate the \ac{los} channel map $g_{k}(\mathbf{x})$
for the \ac{los} region $\mathcal{D}_{k}$ in the neighborhood of
$\mathbf{x}=\mathbf{c}_{0}$:
\begin{equation}
\hat{g}(\mathbf{x})=\alpha+\boldsymbol{\beta}^{\text{T}}(\mathbf{x}-\mathbf{c}_{0})\label{eq:linear-model-channel}
\end{equation}
where $\mathbf{\boldsymbol{\theta}}\triangleq[\alpha,\boldsymbol{\beta}^{\text{T}}]^{\text{T}}$,
in which, $\alpha\in\mathbb{R}$ and $\boldsymbol{\beta}=[\beta_{1},\beta_{2},\beta_{3}]^{\text{T}}\in\mathbb{R}^{3}$
are channel parameters to be estimated based on the \ac{los} measurements
modeled in (\ref{eq:measurement_with_noise}).

For each node $k$, let $\{(\mathbf{x}_{m},y_{m}),m=1,2,\dots,M\}$
be the set of measurements where all the measurements are assumed
taken in the \ac{los} case and $y_{m}$ is the noisy measurement
(\ref{eq:measurement_with_noise}) \ac{wrt} to user $k$ at $\mathbf{x}_{m}=\mathbf{x}(t_{m})$,
where the measurement noise $\xi_{m}$ is assumed to be independent.
Note that one can easily differentiate \ac{los} from \ac{nlos} measurements
by simply tracking the received signal strength. A least-squares solution
of $\mathbf{\boldsymbol{\theta}}$ can be derived as \cite{FanGij:B96,SunChe:J22}
\begin{equation}
\hat{\boldsymbol{\theta}}\triangleq\left[\begin{array}{c}
\hat{\alpha}\\
\hat{\boldsymbol{\beta}}
\end{array}\right]=\left(\tilde{\boldsymbol{\mathbf{X}}}^{\text{T}}\tilde{\boldsymbol{\mathbf{X}}}\right)^{-1}\tilde{\boldsymbol{\mathbf{X}}}^{\text{T}}\mathbf{y}\label{eq:least-square-solution}
\end{equation}
where $\mathbf{y}=[y_{1},y_{2},\dots,y_{M}]^{\text{T}}$, and $\tilde{\mathbf{X}}=[\boldsymbol{1},\mathbf{X}]$
in which, $\boldsymbol{1}$ is a column vector of all $1$s, and $\mathbf{\mathbf{X}}$
is an $M\times3$ matrix with the $m$th row given by $(\mathbf{x}_{m}-\mathbf{c}_{0})^{\text{T}}$.

As a result, for each user $k$, the channel gain and gradient at
$\mathbf{x}=\mathbf{c}_{0}$ for computing $\mathcal{V}(\mathbf{x})$
in (\ref{eq:solution-to-local-approximation-equipotential-point})
are estimated as $g_{k}(\mathbf{c}_{0})=\hat{\alpha}$ and $\nabla g_{k}(\mathbf{c}_{0})=\hat{\bm{\beta}}$.

\subsection{Optimal Measurement Pattern}

\label{subsec:Optimal-Measurement-Pattern}

From the least-squares solution $\hat{\boldsymbol{\theta}}$ in (\ref{eq:least-square-solution}),
the measurement pattern $\tilde{\boldsymbol{\mathbf{X}}}$ affects
the construction performance. We optimize the measurement pattern
$\tilde{\boldsymbol{\mathbf{X}}}$ via analyzing the estimation error
$\boldsymbol{\theta}^{\text{(e)}}=\hat{\boldsymbol{\theta}}-\boldsymbol{\theta}$.
\begin{thm}[Minimum variance]
\label{thm:variance-minimization}Given $d(\mathbf{x}_{m},\mathbf{c}_{0})\leq r_{1}$
for all $m=1,2,\dots,M$, the lower bound of the variance of the estimation
error $\boldsymbol{\theta}^{\text{{\rm (e)}}}$ is given by 
\begin{equation}
\text{{\rm tr}}\left\{ \mathbb{V}\{\boldsymbol{\theta}^{\text{{\rm (e)}}}\}\right\} \geq\frac{\sigma^{2}}{M}+\frac{9\sigma^{2}}{Mr_{1}^{2}}\label{eq:lower-bound-trace-of-variance}
\end{equation}
with equality achieved when the following conditions are satisfied
for all coordinates $j,j'\in\{1,2,3\}$: (i)$\sum_{m=1}^{M}(x_{mj}-c_{0j})=0$;
(ii) $\sum_{m=1}^{M}(x_{mj}-c_{0j})(x_{mj'}-c_{0j'})=0$, for $j\neq j'$;
(iii) $\sum_{m=1}^{M}(x_{mj}-c_{0j})^{2}=Mr_{1}^{2}/3$.
\end{thm}
\begin{proof}
See Appendix~\ref{sec:Proof-of-Theorem-variance-minimization}.
\end{proof}
Theorem~\ref{thm:variance-minimization} indicates that to achieve
a minimum variance, the optimal measurement locations $\mathbf{x}_{m}$
should be distributed in an even and symmetric way about $\mathbf{c}_{0}$.
One possible realization is to distribute the measurements uniformly
on a ball with radius $Mr_{1}^{2}/3$ centered at $\mathbf{c}_{0}$.

The variance lower bound (\ref{eq:lower-bound-trace-of-variance})
consists of two terms where the first term $\sigma^{2}/M$ represents
the error of the estimated $g(\mathbf{c}_{0})=\alpha$, and the second
term represents the error of the estimated $\boldsymbol{\beta}$,
{\em i.e.}, the gradient $\nabla g(\mathbf{c}_{0})$ of the channel
gain model (see Appendix~\ref{sec:Proof-of-Theorem-MSE-Channel-Gain}).
Both terms decrease as the number of measurements $M$ increases.
In addition, the error of $\boldsymbol{\beta}$ decreases as the measurement
radius $r_{1}$ increases.

To derive the optimal measurement range $r_{1}$, we analyze the \ac{mse}
of the locally constructed channel $\hat{g}(\mathbf{x})$ as follows.
\begin{thm}[MSE of the estimated channel gain]
\label{thm:mse-of-estimated-channel-gain}Suppose that $\{\mathbf{x}_{m}\}$
for $m=1,2,\dots,M$ satisfy conditions (i)-(iii) in Theorem~\ref{thm:variance-minimization}.
For location $\mathbf{x}$ with $r_{0}=d(\mathbf{x},\mathbf{c}_{0})$,
the \ac{mse} of the locally constructed channel $\hat{g}(\mathbf{x})$
is upper bounded as
\begin{align}
\mathbb{E}\left\{ \left(\hat{g}(\mathbf{x})-g(\mathbf{x})\right)^{2}\right\} \leq & \frac{\sigma^{2}}{M}\left(1+\frac{3r_{0}^{2}}{r_{1}^{2}}\right)\label{eq:upper-bound-local-channel-model}\\
 & +\frac{L_{g}^{2}}{4}\left(r_{1}^{2}+3r_{0}r_{1}+r_{0}^{2}\right)^{2}.\nonumber 
\end{align}
In addition, if $\mathbf{x}_{m}$ also satisfies $\|\mathbf{x}_{m}-\mathbf{c}_{0}\|^{2}=r_{1}^{2}$,
and if $g(\mathbf{x})$ can be locally approximated by a second order
model,\footnote{Under free-space propagation, the parameter $L'_{g}$ in Theorem~\ref{thm:mse-of-estimated-channel-gain}
can be calculated as $3.5\times10^{-3}$ at a propagation distance
of $50$ meters.} {\em i.e.}, $g(\mathbf{x})\approx g(\mathbf{c}_{0})+\bm{\beta}^{T}(\mathbf{x}-\mathbf{c}_{0})+L'_{g}\|\mathbf{x}-\mathbf{c}_{0}\|^{2}/2$,
then the \ac{mse} is approximately by 
\begin{align}
\mathbb{E}\left\{ \left(\hat{g}(\mathbf{x})-g(\mathbf{x})\right)^{2}\right\}  & \approx\frac{\sigma^{2}}{M}\left(1+\frac{3r_{0}^{2}}{r_{1}^{2}}\right)+\frac{(L_{g}')^{2}}{4}\left(r_{1}^{2}+r_{0}^{2}\right)^{2}.\label{eq:upper-bound-local-channel-model-1}
\end{align}
\end{thm}
\begin{proof}
See Appendix~\ref{sec:Proof-of-Theorem-MSE-Channel-Gain}.
\end{proof}
The result in Theorem~\ref{thm:mse-of-estimated-channel-gain} demonstrates
a trade-off in the measurement range $r_{1}$. When the channel model
$g(\mathbf{x})$ has a small Lipschitz constant $L_{g}$ in (\ref{eq:Lipschitz-g-upper})
and (\ref{eq:Lipschitz-g-lower}), corresponding to a small curvature,
a large range $r_{1}$ is preferred for a small variance in (\ref{eq:upper-bound-local-channel-model}),
leading to a small \ac{mse}. When the channel model $g(\mathbf{x})$
has a large $L_{g}$, a small range $r_{1}$ is preferred, because
the linear model (\ref{eq:linear-model-channel}) becomes less accurate
in the range $r_{0}$ for a large $L_{g}$.

The \ac{mse} upper bound in Theorem~\ref{thm:mse-of-estimated-channel-gain}
implies an optimal choice of the measurement range $r_{1}$. By minimizing
the approximated \ac{mse} (\ref{eq:upper-bound-local-channel-model-1}),
Table~\ref{tab:optimal_r} shows some numerical examples on the optimal
choice of $r_{1}$ under different values of $M$ and $r_{0}$, where
$\sigma=5$~dB and the parameter $L_{g}'$ is obtained from a free-space
propagation model evaluated for a neighborhood at a propagation distance
of $50$ meters.

\begin{table}
\caption{\label{tab:optimal_r}Optimal choice of $r_{1}$ {[}meter{]} under
different $r_{0}$ {[}meter{]} and $M$}

\renewcommand{\arraystretch}{1.2}
\centering{}%
\begin{tabular}{>{\centering}m{0.13\columnwidth}|>{\centering}m{0.15\columnwidth}|>{\centering}m{0.15\columnwidth}|>{\centering}m{0.15\columnwidth}|>{\centering}m{0.16\columnwidth}}
\hline 
\centering{}$r_{1}^{*}$ & \centering{}$M=40$ & \centering{}$M=60$ & \centering{}$M=80$ & \centering{}$M=100$\tabularnewline
\hline 
\centering{}$r_{0}=10$ & \centering{}$17$ & \centering{}$16$ & \centering{}$15$ & \centering{}$14$\tabularnewline
\hline 
\centering{}$r_{0}=20$ & \centering{}$20$ & \centering{}$18$ & \centering{}$17$ & \centering{}$16$\tabularnewline
\hline 
\centering{}$r_{0}=30$ & \centering{}$21$ & \centering{}$20$ & \centering{}$18$ & \centering{}$17$\tabularnewline
\hline 
\end{tabular}
\end{table}

\subsection{Measurement Trajectory Design}

\label{subsec:measurement-Trajectory-Design}

Here, we construct the measurement trajectory $\mathbf{x}_{\text{r}}(t)$
to meet conditions (i)--(iii) in Theorem~\ref{thm:variance-minimization}
for achieving a small error in locally constructing $\bm{g}(\mathbf{x})$
along the search $\mathbf{x}_{\text{s}}(t)$.

For the ease of elaboration, consider that the exploration direction
is given by $\dot{\mathbf{x}}_{\text{s}}=\mathbf{s}=[0,1,0]^{\text{T}}$
for a piece of trajectory $\mathbf{x}_{\text{s}}(t)$ centered at
$\mathbf{c}_{0}=[0,0,0]^{\text{T}}$ as illustrated in Fig.~\ref{fig:sphere-and-cylinder}.
One can construct a horizontal cylinder with length $2r_{1}/\sqrt{3}$
and radius $\sqrt{2/3}r_{1}$, where the measurement range $r_{1}$
is chosen according to Theorem~\ref{thm:mse-of-estimated-channel-gain}
for a good construction performance within a one-step exploration
range $r_{0}$ obtained while solving the equipotential surface tracking
problem (\ref{eq:distance-minimization-problem}). The orientation
of the cylinder is given by $\mathbf{s}$. It follows that if one
uniformly samples along the circumferences of the top and bottom bases
of the cylinder as shown by the blue dots in Fig.~\ref{fig:sphere-and-cylinder},
the resulting sampling locations $\{\mathbf{x}_{m}\}$ satisfy all
conditions in Theorem~\ref{thm:variance-minimization}.

Note that the above trajectory only visits two distinct positions
in the direction $\mathbf{s}$. Alternatively, along the search trajectory
$\mathbf{x}_{\text{s}}(t)=[0,vr_{1}(t-M/2),0]^{\text{T}}$ for a speed
parameter $v$, one may consider an \emph{alternating spiral trajectory}
$\mathbf{x}(t)=\mathbf{x}_{\text{s}}(t)+\mathbf{x}_{\text{r}}(t)$
with $\mathbf{x}_{\text{r}}(t)=[x_{\text{r}1}(t),0,x_{\text{r}3}(t)]^{\text{T}}$,
where
\begin{equation}
\begin{cases}
x_{\text{r}1}(t)=\sqrt{2/3}r_{1}\cos(\omega t)\\
x_{\text{r}3}(t)=\sqrt{2/3}r_{1}\sin(\omega t)(-1)^{\lfloor\omega t/(2\pi)\rfloor}
\end{cases}\label{eq:common-spiral-trajectory-equation}
\end{equation}
where $\omega=4k\pi/M$ and $v=2/\sqrt{M^{2}-1}$, with $k$ being
a natural number, typically $k=1$ (see Fig.~\ref{fig:sphere-and-cylinder}
(a)). One can easily verify that if we sample at $t=m-1/2$, {\em i.e.},
$\mathbf{x}_{m}=\mathbf{x}(m-1/2)$, for $m=1,2,\dots,M$, then conditions
(i)--(iii) in Theorem~\ref{thm:variance-minimization} are satisfied.
As a result, the \ac{mse} of the locally reconstructed channel $\hat{g}_{k}(\mathbf{x})$
is upper bounded by (\ref{eq:upper-bound-local-channel-model}).
\begin{figure}
\begin{centering}
\subfigure[]{\includegraphics[width=0.5\columnwidth]{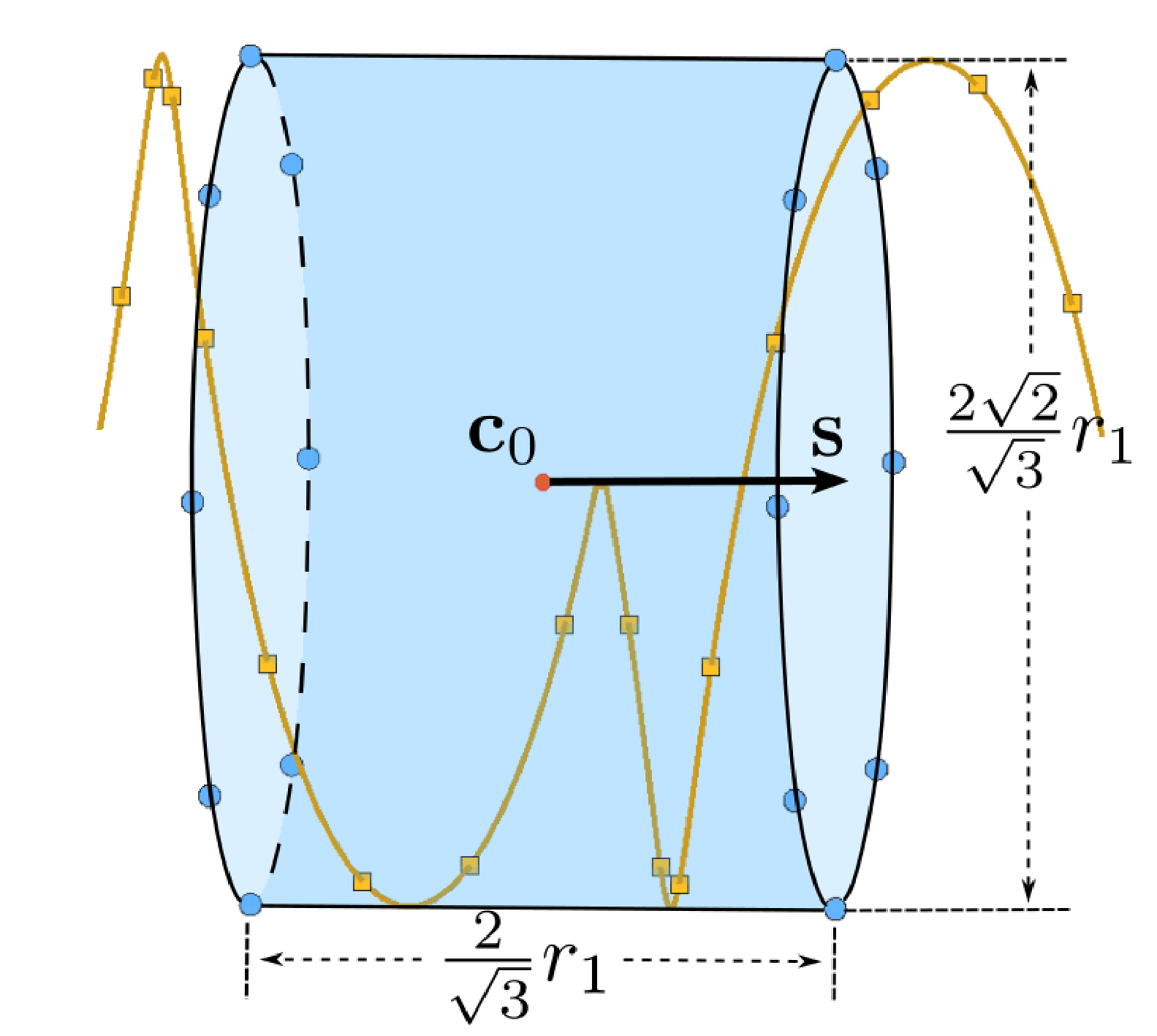}}\subfigure[]{\includegraphics[width=0.5\columnwidth]{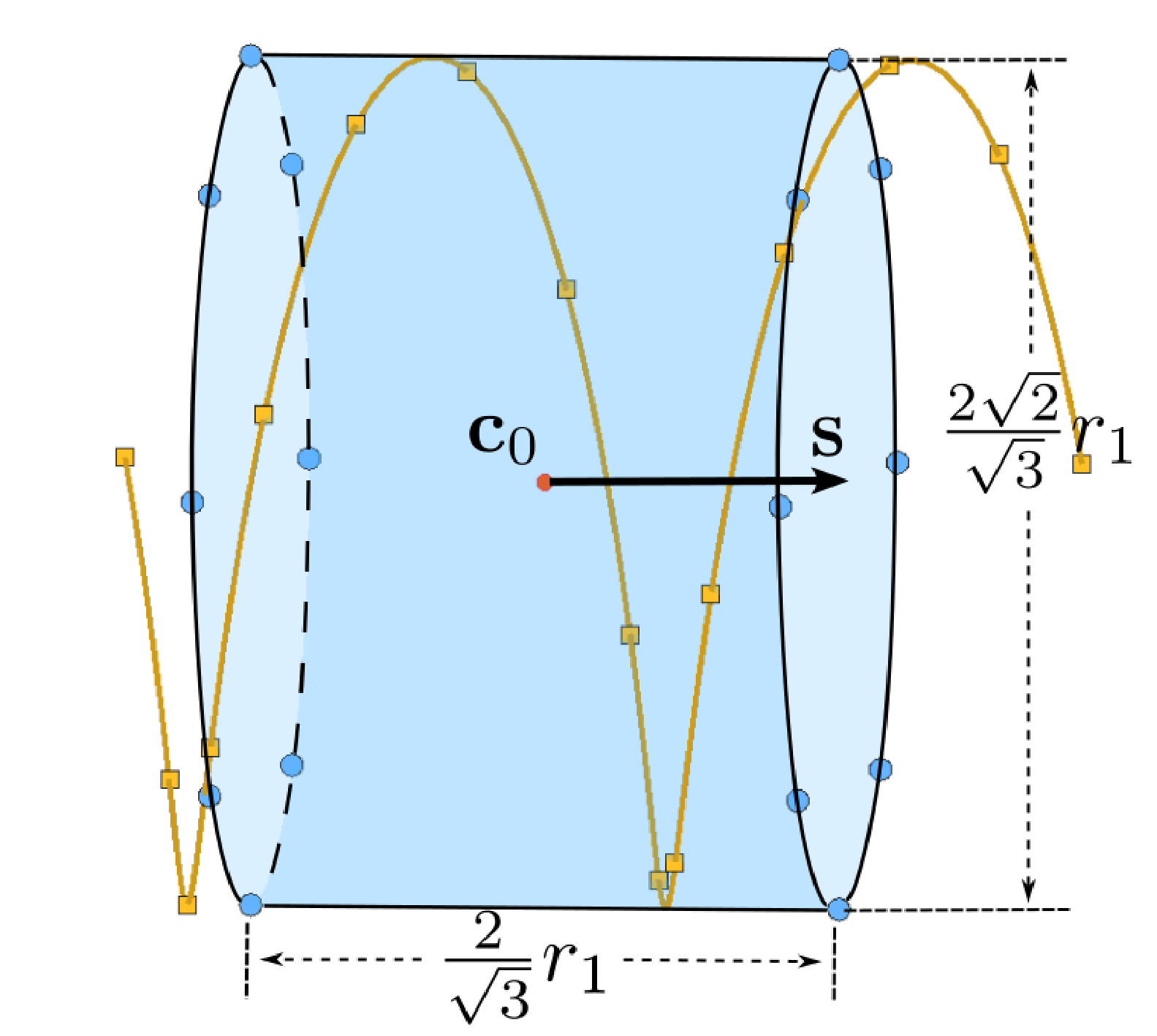}}
\par\end{centering}
\centering{}\caption{\label{fig:sphere-and-cylinder}(a) An alternating spiral trajectory
(orange dots) that satisfies conditions~(i)--(iii) in Theorem~\ref{thm:variance-minimization}.
(b) A spiral trajectory that satisfies conditions~(i) and (iii) in
Theorem~\ref{thm:variance-minimization} and it is smooth.}
\end{figure}

\section{Trajectory Design for Searching Optimal LOS Position on the Equipotential
Surface}

\label{sec:trajectory-design-for-searching-optimal-LOS-on-ES}

When the search is constrained on the equipotential surface $\mathcal{S}$
where it holds that $F(\bm{g}(\mathbf{x}))=f_{0}(g_{0}(\mathbf{x}))-F_{\text{u}}(\bm{g}_{\text{u}}(\mathbf{x}),\mathbf{p}^{*}(\mathbf{x}))=0$,
the original problem $\mathscr{P}$ becomes 
\begin{equation}
\begin{aligned}\mathop{\mbox{maximize}}\limits _{\mathbf{x},\mathbf{p}} & \quad F_{\text{u}}(\bm{g}_{\text{u}}(\mathbf{x}),\mathbf{p})\\
\mathop{\mbox{subject to}} & \quad\mathbf{x}\in\mathcal{S}\cap\tilde{\mathcal{D}},\\
 & \quad H_{n}(\bm{g}_{\text{u}}(\mathbf{x}),\mathbf{p})\leq0,n=1,2,\dots,N.
\end{aligned}
\label{general_formula-resource-allocation}
\end{equation}

It is very challenging to handle the constraint $\mathbf{x}\in\mathcal{S}$,
especially when the analytical form of $\mathcal{S}$ is not available.
Some classical approaches may consider projection-type algorithms,
where the position $\mathbf{x}(t)$ is projected back to $\mathcal{S}$
whenever $\mathbf{x}(t)$ is off $\mathcal{S}$, {\em e.g.}, via
the trajectory developed in Section~\ref{sec:Trajectory-Design-for-Equipotential-Surface}.
However, such a projection-type search is not suitable for \ac{uav}
trajectory design except for an initialization phase, because frequent
projections may cost a large amount of maneuver energy for the \ac{uav}.
Therefore, it is desired that the \ac{uav} only moves on $\mathcal{S}$.

In this section, we develop a search trajectory $\mathbf{x}_{\text{s}}(t)$
sticking on the equipotential plane $\mathcal{S}$ without projections.
The challenge is that a perturbation on $\mathbf{x}_{\text{s}}(t)$
may change the channel gain $g_{k}(\mathbf{x}_{\text{s}})$, and hence,
the optimal resource allocation $\mathbf{p}^{*}(\mathbf{x}_{\text{s}})$,
possibly resulting in $F(\bm{g}(\mathbf{x}_{\text{s}}))\neq0$. We
tackle this challenge via the perturbation theory.

\subsection{Trajectory on the Equipotential Surface}

\label{subsec:Trajectory-on-the-equipotential-surface}

We simplify the elaboration by temporarily ignoring $\mathbf{x}_{\text{r}}(t)$,
and hence, $\mathbf{x}(t)=\mathbf{x}_{\text{s}}(t)$. Start from a
position $\text{\ensuremath{\mathbf{x}}}(0)\in\mathcal{S}$, which
can be obtained from the trajectory in Section~\ref{sec:Trajectory-Design-for-Equipotential-Surface}.
To investigate the property of the trajectory $\text{\ensuremath{\mathbf{x}}}(t)\in\mathcal{S}$,
we analyze the optimality of (\ref{general_formula-resource-allocation})
via the Lagrangian approach as follows.

For the problem (\ref{general_formula-resource-allocation}), denote
the Lagrangian function as 
\[
L(\mathbf{p},\boldsymbol{\lambda};\bm{g}_{\text{u}}(\mathbf{x}))=F_{\text{u}}(\mathbf{p};\bm{g}_{\text{u}}(\mathbf{x}))-{\textstyle \sum_{n=1}^{N}}\lambda_{n}H_{n}(\mathbf{p};\bm{g}_{\text{u}}(\mathbf{x}))
\]
where $\boldsymbol{\lambda}=[\lambda_{1},\lambda_{2},\dots,\lambda_{N}]^{\text{T}}$.
The \ac{kkt} conditions are written as 
\begin{equation}
\mathbf{J}(\mathbf{p}(\mathbf{x}),\bm{\lambda}(\mathbf{x});\bm{g}_{\text{u}}(\mathbf{x}))\triangleq\left[\begin{array}{c}
\nabla_{\mathbf{p}}L(\mathbf{p},\boldsymbol{\lambda};\bm{g}_{\text{u}}(\mathbf{x}))\\
\lambda_{1}H_{1}(\mathbf{p};\bm{g}_{\text{u}}(\mathbf{x}))\\
\lambda_{2}H_{2}(\mathbf{p};\bm{g}_{\text{u}}(\mathbf{x}))\\
\vdots\\
\lambda_{N}H_{N}(\mathbf{p};\bm{g}_{\text{u}}(\mathbf{x}))
\end{array}\right]=0\label{eq:KKT-condition-0}
\end{equation}
together with $\lambda_{n}\geq0$ and $H_{n}(\mathbf{p}(\mathbf{x});\bm{g}_{\text{u}}(\mathbf{x}))\leq0$
for all $n$. It is known that for a strictly convex problem, there
is a unique solution $\{\mathbf{p}^{*},\boldsymbol{\lambda}^{*}\}$
to $\mathbf{J}(\mathbf{p}(\mathbf{x}),\bm{\lambda}(\mathbf{x});\bm{g}_{\text{u}}(\mathbf{x}))=0$
while satisfying $\lambda_{n}\geq0$ and $H_{n}(\mathbf{p}(\mathbf{x});\bm{g}_{\text{u}}(\mathbf{x}))\leq0$.

Since (\ref{eq:KKT-condition-0}) and $F(\bm{g}(\mathbf{x}))=0$ are
expected to be satisfied for all $\mathbf{x}(t)$, $t\geq0$, we must
have
\begin{equation}
\begin{cases}
\frac{\text{d}}{\text{d}t}\mathbf{J}(\mathbf{p}(\mathbf{x}(t)),\bm{\lambda}(\mathbf{x}(t));\bm{g}_{\text{u}}(\mathbf{x}(t))) & =0\\
\frac{\text{d}}{\text{d}t}F(\bm{g}(\mathbf{x}(t))) & =0
\end{cases}\label{eq:J-0-F-0}
\end{equation}
which leads to 
\begin{equation}
\begin{cases}
\nabla_{\mathbf{p}}\mathbf{J}^{\text{T}}\dot{\mathbf{p}}^{*}+\nabla_{\boldsymbol{\lambda}}\mathbf{J}^{\text{T}}\dot{\bm{\lambda}}^{*}+\nabla_{\bm{g}_{\text{u}}}\mathbf{J}^{\text{T}}\nabla\bm{g}_{\text{u}}^{\text{T}}\dot{\mathbf{x}} & =0\\
\nabla_{\mathbf{p}}F^{\text{T}}\dot{\mathbf{p}}^{*}+\nabla_{\bm{g}}F^{\text{T}}\nabla\bm{g}^{\text{T}}\dot{\mathbf{x}} & =0
\end{cases}\label{eq:J-0-F-0-equations}
\end{equation}
where we use the notation $\dot{\mathbf{p}}^{*}=\text{d}\mathbf{p}^{*}(t)/\text{d}t$,
$\dot{\bm{\lambda}}^{*}=\text{d}\bm{\lambda}^{*}(t)/\text{d}t$, and
$\dot{\mathbf{x}}=\text{d}\mathbf{x}(t)/\text{d}t$.

The above dynamical system (\ref{eq:J-0-F-0-equations}) specifies
a motion $\dot{\mathbf{x}}$ on the equipotential surface $\mathcal{S}$
with two spatial degrees of freedom. Suppose that the search is further
constrained on a plane that intersects with $\mathcal{S}$. Denote
$\mathbf{q}$ as the normal vector of the search plane, {\em i.e.},
$\mathbf{q}^{\text{T}}\dot{\mathbf{x}}=0$. Then, the dynamic of the
trajectory satisfies 
\begin{align}
\left[\begin{array}{ccc}
\nabla_{\mathbf{p}}\mathbf{J}^{\text{T}} & \nabla_{\boldsymbol{\lambda}}\mathbf{J}^{\text{T}} & \nabla_{\bm{g}_{\text{u}}}\mathbf{J}^{\text{T}}\nabla\bm{g}_{\text{u}}^{\text{T}}\\
\nabla_{\mathbf{p}}F^{\text{T}} & \mathbf{0} & \nabla_{\bm{g}}F^{\text{T}}\nabla\bm{g}^{\text{T}}\\
\mathbf{0} & \mathbf{0} & \mathbf{q}^{\text{T}}\\
\mathbf{0} & \mathbf{0} & \mathbf{v}^{\text{T}}
\end{array}\right]\left[\begin{array}{c}
\dot{\mathbf{p}}^{*}\\
\dot{\bm{\lambda}}^{*}\\
\dot{\mathbf{x}}
\end{array}\right] & =\left[\begin{array}{c}
\mathbf{0}\\
1
\end{array}\right]\label{eq:J-0-F-0-matrix}
\end{align}
where the vector $\mathbf{v}$ can be randomly chosen and the last
equality $\mathbf{v}^{\text{T}}\dot{\mathbf{x}}=1$ is to avoid the
trivial solution $\dot{\mathbf{x}}=\mathbf{0}$,\footnote{The physical meaning of $\mathbf{v}$ is to control the speed $\dot{\mathbf{x}}$
projected on the direction $\mathbf{v}$. For example, $\mathbf{v}=[0,0,1]^{\text{T}}$
specifies the speed $\dot{x}_{3}=1$.} and therefore, the system of equations~(\ref{eq:J-0-F-0-matrix})
is completely determined.

To derive a closed-form expression for $\dot{\mathbf{x}}$, let $\mathbf{A}_{1}=[\begin{array}{cc}
\nabla_{\mathbf{p}}\mathbf{J}^{\text{T}} & \nabla_{\boldsymbol{\lambda}}\mathbf{J}^{\text{T}}\end{array}]$, $\mathbf{A}_{2}=\nabla_{\bm{g}_{\text{u}}}\mathbf{J}^{\text{T}}\nabla\bm{g}_{\text{u}}^{\text{T}}$,
\begin{equation}
\mathbf{A}_{3}=\left[\begin{array}{cc}
\nabla_{\mathbf{p}}F^{\text{T}} & \mathbf{0}\\
\mathbf{0} & \mathbf{0}
\end{array}\right],\quad\mathbf{A}_{4}(\mathbf{q})=\left[\begin{array}{c}
\nabla F^{\text{T}}\nabla\bm{g}^{\text{T}}\\
\mathbf{q}^{\text{T}}\\
\mathbf{v}^{\text{T}}
\end{array}\right]\label{eq:equation-A4}
\end{equation}
and $\mathbf{e}_{3}=[0,0,1]^{\text{T}}$. Using the block matrix inversion
lemma, the dynamic $\dot{\mathbf{x}}$ as the solution to (\ref{eq:J-0-F-0-matrix})
can be derived as 
\begin{equation}
\dot{\mathbf{x}}=\mathcal{A}(\mathbf{x};\mathbf{q})\triangleq\left(\mathbf{A}_{4}(\mathbf{q})-\mathbf{A}_{3}\mathbf{A}_{1}^{-1}\mathbf{A}_{2}\right)^{-1}\mathbf{e}_{3}.\label{eq:solution-dx-dt}
\end{equation}

\begin{figure}
\begin{centering}
\subfigure[]{\includegraphics[width=0.5\columnwidth]{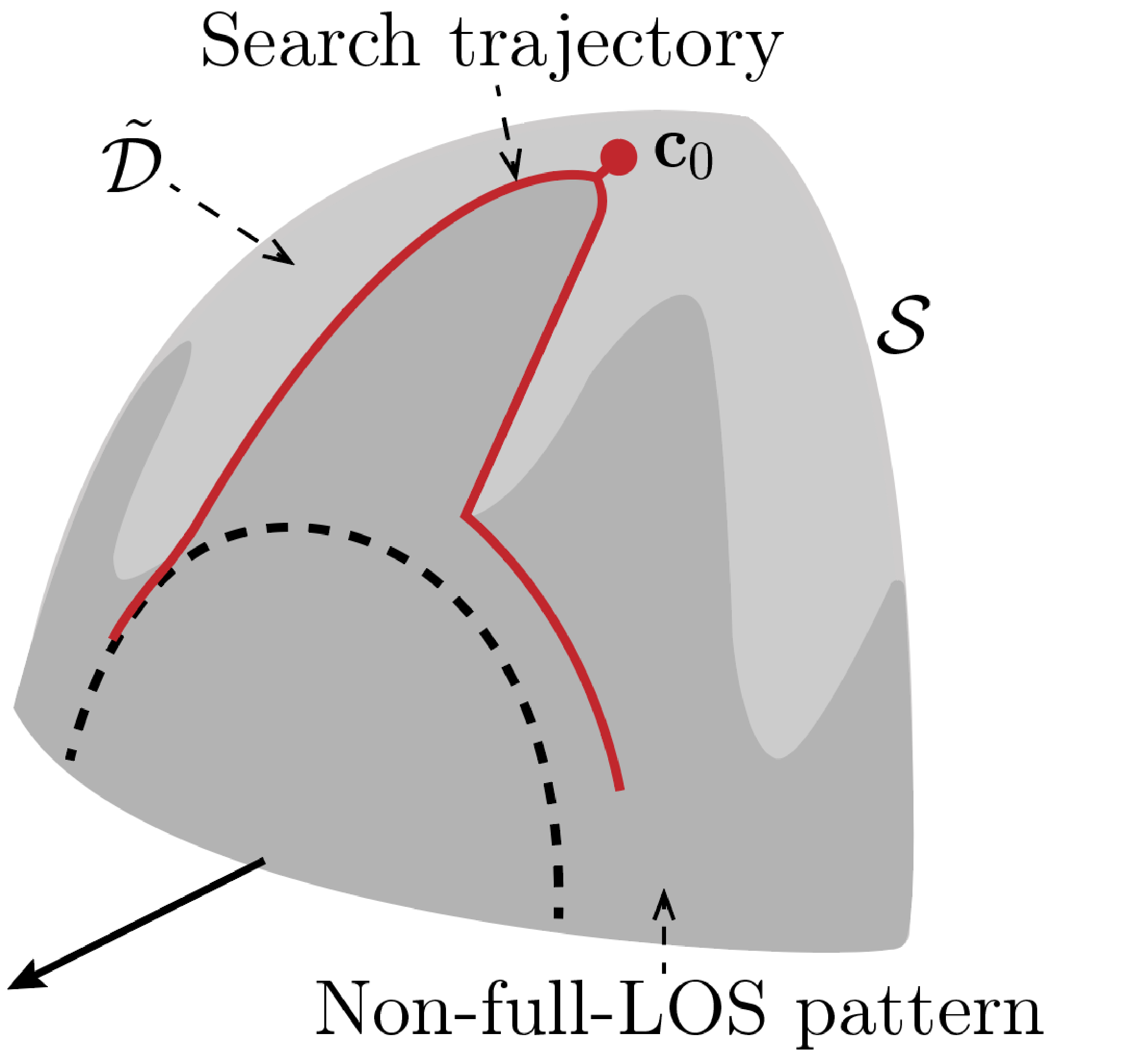}}\subfigure[]{\includegraphics[width=0.5\columnwidth]{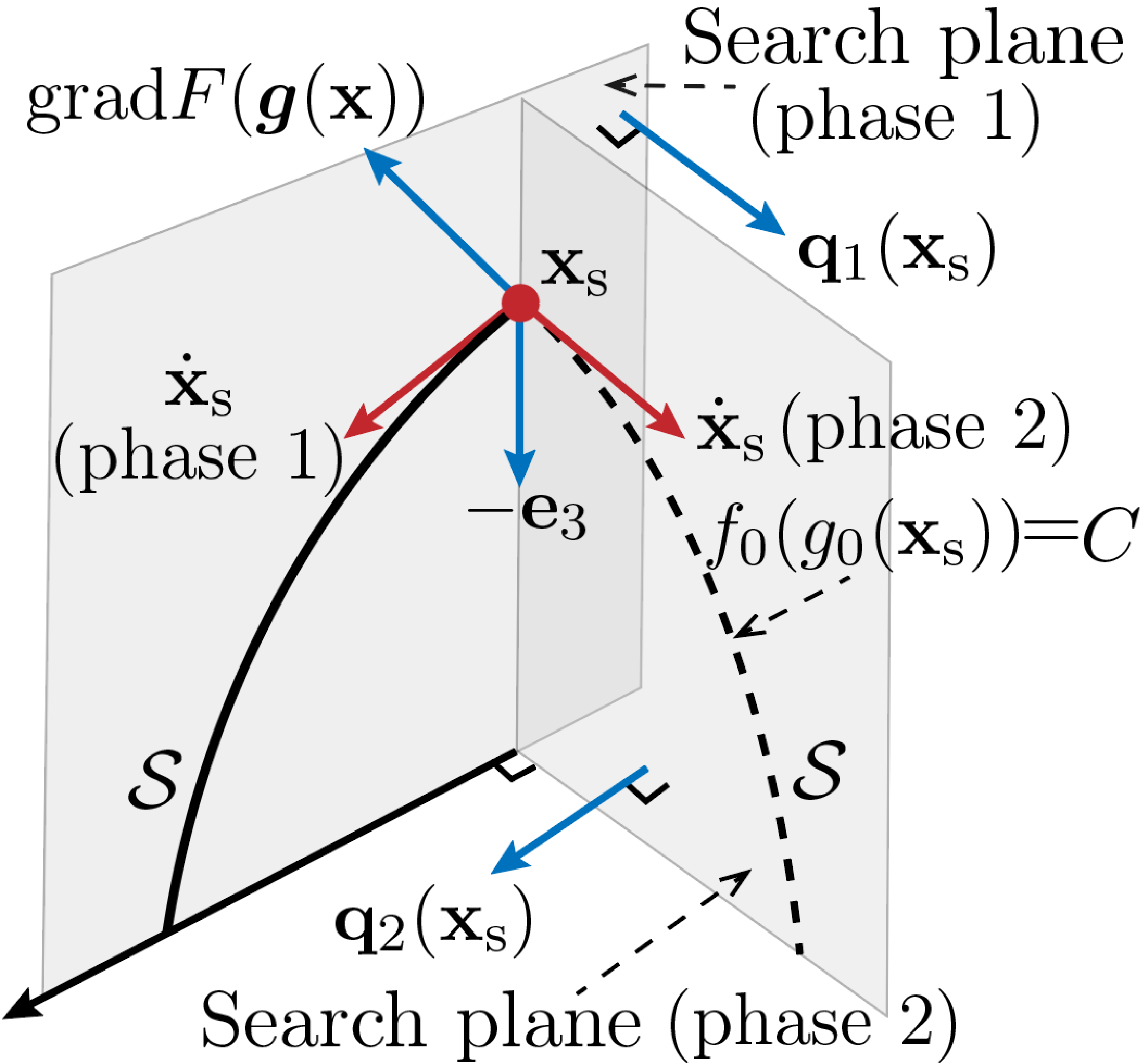}}
\par\end{centering}
\centering{}\caption{\label{fig:Trajectory-on-the-Equipotential-Surface}(a) LOS discovery
trajectory on the equipotential surface. (b) Search directions in
Phase~$1$ and $2$.}
\end{figure}

\subsection{\ac{los} Discovery on the Equipotential Surface}

\label{subsec:LOS-aware-Search-Trajectory}

The dynamical system (\ref{eq:J-0-F-0-matrix}) requires specifying
a search direction determined by the normal vector\textbf{ $\mathbf{q}$}
in (\ref{eq:J-0-F-0-matrix}). We present a search strategy to adaptively
determine the search direction via the normal vector $\mathbf{q}$
to discover the best \ac{los} opportunity that solves problem (\ref{general_formula-resource-allocation}).
The strategy is based on the following two properties.

First, consider the curve $f_{0}(g_{0}(\mathbf{x}))=C$ on the equipotential
surface $\mathbf{x}\in\mathcal{S}$, as illustrated by the dashed
curve in Fig.~\ref{fig:Trajectory-on-the-Equipotential-Surface}(a).
It follows that the lower the curve, {\em i.e.}, with a smaller radius,
the larger the objective value $f_{0}(g_{0}(\mathbf{x}))=F_{\text{u}}(\bm{g}_{\text{u}}(\mathbf{x}),\mathbf{p}^{*}(\mathbf{x}))$,
because $\mathbf{x}$ is closer to the \ac{bs} resulting in a larger
channel gain $g_{0}(\mathbf{x})$. Second, there is an upward invariant
property for the full LOS region $\tilde{\mathcal{D}}$ due to the
environment model in Section~\ref{subsec:Environment-and-Propagation-model},
{\em i.e.}, if $\mathbf{x}\notin\tilde{\mathcal{D}}$, the locations
below $\mathbf{x}$ are also in \ac{nlos}.

These observations inspire the following search strategy for an \ac{los}
discovery trajectory $\mathbf{x}_{\text{s}}(t)$.
\begin{itemize}
\item \textbf{Phase 1:} if $\mathbf{x}_{\text{s}}(t)\in\mathcal{S}\cap\tilde{\mathcal{D}}$,
then one should decrease the altitude of $\mathbf{x}_{\text{s}}(t)$
to discover a larger $f_{0}(g_{0}(\mathbf{x}_{\text{s}}(t)))$. Thus,
the search direction $\dot{\mathbf{x}}$ should lie on the tangent
plane of the equipotential surface $\mathcal{S}$ and be as close
to the downward vector $-\mathbf{e}_{3}=[0,0,-1]^{\text{T}}$ as possible.
Analytically, the normal vector $\mathbf{q}$ in (\ref{eq:J-0-F-0-matrix})
of the search plan that contains $\dot{\mathbf{x}}$ should be orthogonal
to both the normal vector $\text{grad}\,F(\bm{g}(\mathbf{x}_{\text{s}}))=\nabla F(\bm{g}(\mathbf{c}_{0}))^{\text{T}}\mathbf{G}(\mathbf{c}_{0})$
for $\mathcal{S}$ and the downward vector $-\mathbf{e}_{3}$, {\em i.e.},
\begin{equation}
\mathbf{q}_{1}(\mathbf{x}_{\text{s}})=\frac{\text{grad}\,F(\bm{g}(\mathbf{x}_{\text{s}}))\times(-\mathbf{e}_{3})}{\|\text{grad}\,F(\bm{g}(\mathbf{x}_{\text{s}}))\times(-\mathbf{e}_{3})\|_{2}}\label{eq:closed-form-q-in-LOS}
\end{equation}
where $\mathbf{a}\times\mathbf{b}$ denotes the cross product of $\mathbf{a}$
and $\mathbf{b}$.
\item \textbf{Phase 2:} if $\mathbf{x}_{\text{s}}(t)\in\mathcal{S}$ but
$\mathbf{x}\notin\tilde{\mathcal{D}}$, one explores the equipotential
surface following the curve $f_{0}(g_{0}(\mathbf{x}_{\text{s}}(t)))=C$
to discover an \ac{los} opportunity. Analytically, the curve $f_{0}(g_{0}(\mathbf{x}_{\text{s}}(t)))=C$
satisfies the following \ac{ode} $\frac{\text{d}}{\text{d}t}f_{0}(g_{0}(\mathbf{x}_{\text{s}}(t)))=\nabla f_{0}\nabla g_{0}^{\text{T}}\dot{\mathbf{x}}_{\text{s}}=0$,
and thus, the normal vector $\mathbf{q}$ is given by 
\begin{equation}
\mathbf{q}_{2}(\mathbf{x}_{\text{s}})=\text{\ensuremath{\frac{\nabla f_{0}(g_{0}(\mathbf{x}_{\text{s}}))\nabla g_{0}(\mathbf{x}_{\text{s}})^{\text{T}}}{\|\nabla f_{0}(g_{0}(\mathbf{x}_{\text{s}}))\nabla g_{0}(\mathbf{x}_{\text{s}})^{\text{T}}\|_{2}}}}.\label{eq:closed-form-q2-in-NLOS}
\end{equation}
\end{itemize}

\subsection{Superposed Trajectory via ODEs}

\label{subsec:Integrated-Trajectory-Design-Under-Unknown}

\subsubsection{The Superposed Trajectory}

\label{subsec:The-Superposed-Trajectory}

Here, we combine the search trajectory $\mathbf{x}_{\text{s}}(t)$
with the measurement trajectory $\mathbf{x}_{\text{r}}(t)$ developed
in Section~\ref{subsec:measurement-Trajectory-Design}, assuming
the initial state satisfies $\mathbf{x}_{\text{s}}(0)\in\mathcal{S}$.

First, from (\ref{eq:equipotential-surface-tracking}) and (\ref{eq:solution-dx-dt}),
the combined search trajectory is given by $\dot{\mathbf{x}}_{\text{s}}=\mathcal{A}(\mathbf{x}_{\text{s}}(t);\mathbf{q}(\mathbf{x}_{\text{s}}(t)))+\mu_{\text{v}}\mathcal{V}(\mathbf{x}_{\text{s}}(t))$,
where the first term is to search on the equipotential surface $\mathcal{S}$
according to the two exploration phases (\ref{eq:closed-form-q-in-LOS})
and (\ref{eq:closed-form-q2-in-NLOS}), and the second term is to
track the $\mathcal{S}$ in case $\mathbf{x}_{\text{s}}(t)$ deviates
from it due to implementation issues.

Second, consider the measurement trajectory $\mathbf{x}_{\text{r}}(t)=r[\cos(\omega t),\,0,\,\sin(\omega t)]^{\text{T}}$
developed in Section~\ref{subsec:measurement-Trajectory-Design},
which forms a circle on the plane with a normal vector $\mathbf{e}_{2}=[0,1,0]^{\text{T}}$.
Then, given the search direction $\dot{\mathbf{x}}_{\text{s}}$, one
can construct a rotation matrix $\mathbf{R}(\dot{\mathbf{x}}_{\text{s}})$
that rotates the coordinate system with the reference direction $\mathbf{e}_{2}$
to a new coordinate system with the reference direction $\dot{\mathbf{x}}_{\text{s}}/\|\dot{\mathbf{x}}_{\text{s}}\|_{2}$.
The rotation matrix $\mathbf{R}(\mathbf{s})$ to the reference direction
$\mathbf{s}=[s_{1},s_{2},s_{3}]^{\text{T}}$ is found as 
\[
\mathbf{R}(\mathbf{s})=\mathbf{I}-\frac{1}{\|\mathbf{s}\|_{2}}\left[\begin{array}{ccc}
\frac{s_{1}^{2}}{\|\mathbf{s}\|_{2}+s_{2}} & s_{1} & \frac{-s_{1}s_{3}}{\|\mathbf{s}\|_{2}+s_{2}}\\
-s_{1} & \|\mathbf{s}\|_{2}-s_{2} & -s_{3}\\
\frac{-s_{1}s_{3}}{\|\mathbf{s}\|_{2}+s_{2}} & s_{3} & \frac{s_{3}^{2}}{\|\mathbf{s}\|_{2}+s_{2}}
\end{array}\right].
\]

The dynamical equation for the superposed UAV search trajectory $\mathbf{x}(t)=\mathbf{x}_{\text{s}}(t)+\mathbf{x}_{\text{r}}(t)$
then becomes 
\begin{align}
\dot{\mathbf{x}} & =\dot{\mathbf{x}}_{\text{s}}+\mathbf{R}(\dot{\mathbf{x}}_{\text{s}})\dot{\mathbf{x}}_{\text{r}}+\frac{\text{d}}{\text{d}t}\mathbf{R}(\dot{\mathbf{x}}_{\text{s}})\mathbf{x}_{\text{r}}(t)\label{eq:Dynamical-Equation-Trajectory-x}\\
\dot{\mathbf{x}}_{\text{s}} & =\mathcal{A}(\mathbf{x}_{\text{s}}(t);\mathbf{q}(\mathbf{x}_{\text{s}}(t)))+\mu_{\text{v}}\mathcal{V}(\mathbf{x}_{\text{s}}(t))\label{eq:Dynamical-Equation-Trajectory-xs}
\end{align}
where $\dot{\mathbf{x}}_{\text{r}}=\text{d}\mathbf{x}_{\text{r}}(t)/\text{d}t=r\omega[-\sin(\omega t),\,0,\,\cos(\omega t)]^{\text{T}}$
and $\frac{\text{d}}{\text{d}t}\mathbf{R}(\dot{\mathbf{x}}_{\text{s}})=\nabla\mathbf{R}(\dot{\mathbf{x}}_{\text{s}})(\mathbf{I}_{3}\otimes\ddot{\mathbf{x}}_{\text{s}})$.
Here, the operator $\nabla\mathbf{R}$ gives a matrix with $3\times3$
blocks, where the $(i,j)$th block is $[\frac{\partial R_{ij}}{\partial x_{1}},\frac{\partial R_{ij}}{\partial x_{2}},\frac{\partial R_{ij}}{\partial x_{3}}]$,
$\otimes$ is the Kronecker product, and $\ddot{\mathbf{x}}_{\text{s}}=\text{d}^{2}\mathbf{x}_{\text{s}}/\text{d}t^{2}$.
In (\ref{eq:Dynamical-Equation-Trajectory-x}), the second term creates
a spiral trajectory surrounding the main search route $\mathbf{x}_{\text{s}}(t)$
for collecting the channel measurement data, and the third term generates
the adjustment due to the potential time variation of the search direction
$\dot{\mathbf{x}}_{\text{s}}$.

\subsubsection{Implementation}

The analytical form of $\frac{\text{d}}{\text{d}t}\mathbf{R}(\dot{\mathbf{x}}_{\text{s}})$
in (\ref{eq:Dynamical-Equation-Trajectory-x}) requires the second-order
derivative of the search trajectory $\mathbf{x}_{\text{s}}(t)$, which
is not available. A simple solution is to use numerical approximations
$\frac{\text{d}}{\text{d}t}\mathbf{R}(\dot{\mathbf{x}}_{\text{s}})\approx\frac{1}{\tau}(\mathbf{R}(\dot{\mathbf{x}}_{\text{s}}(t))-\mathbf{R}(\dot{\mathbf{x}}_{\text{s}}(t-\tau)))$
for a small enough $\tau>0$. Alternatively, we find the following
approximations.

First, when the search $\mathbf{x}_{\text{s}}(t)$ remains in Phase~1
or Phase~2 as specified in Section~\ref{subsec:LOS-aware-Search-Trajectory},
we likely have $\ddot{\mathbf{x}}_{\text{s}}\approx\mathbf{0}$, and
thus the term $\frac{\text{d}}{\text{d}t}\mathbf{R}(\dot{\mathbf{x}}_{\text{s}})\mathbf{x}_{\text{r}}(t)$
can be simply ignored. This is because $\mathbf{x}_{\text{s}}(t)$
is a trajectory on the equipotential surface $\mathcal{S}$, and thus
$\frac{\text{d}}{\text{d}t}\mathbf{R}(\dot{\mathbf{x}}_{\text{s}})$
represents the supplementary rotation due to the curvature of $\mathcal{S}$,
which is relatively small in practical regime of interest compared
to the other terms.\footnote{For example, for a balancing problem where the equipotential plane
$\mathcal{S}$ is a sphere under some mild conditions (Proposition
\ref{prop:shape_of_equipotential_surface}), it is known that the
curvature of a sphere with radius $R$ is $1/R^{2}$, which is small
as $R$ is at the order of a hundred meters.} Hence, the dynamical equation for $\mathbf{x}(t)$ in this case is
approximated as
\begin{align}
\dot{\mathbf{x}} & \approx\dot{\mathbf{x}}_{\text{s}}+\mathbf{R}(\dot{\mathbf{x}}_{\text{s}})\dot{\mathbf{x}}_{\text{r}}.\label{eq:Dynamical-Equation-Trajectory-x-NoSwitch}
\end{align}

Second, when the search $\mathbf{x}_{\text{s}}(t)$ needs to switch,
for instance, from Phase~1 to Phase~2 at time $t=t_{1}$, the search
direction needs to switch between $\dot{\mathbf{x}}_{\text{s}}(t_{1}^{-})=\mathcal{A}(\mathbf{x}_{\text{s}}(t_{1}^{-});\mathbf{q}_{1}(\mathbf{x}_{\text{s}}(t_{1}^{-})))$
and $\dot{\mathbf{x}}_{\text{s}}(t_{1}^{+})=\mathcal{A}(\mathbf{x}_{\text{s}}(t_{1}^{+});\mathbf{q}_{2}(\mathbf{x}_{\text{s}}(t_{1}^{+})))$,
assuming $\mathcal{V}(\mathbf{x}_{\text{s}}(t))=0$ for simplicity.
Thus, $\frac{\text{d}}{\text{d}t}\mathbf{R}(\dot{\mathbf{x}}_{\text{s}})$
does not exist as $\ddot{\mathbf{x}}_{\text{s}}$ does not exist,
since $\dot{\mathbf{x}}_{\text{s}}(t_{1}^{-})\neq\dot{\mathbf{x}}_{\text{s}}(t_{1}^{+})$.
To circumvent this issue, a transition phase $t\in(t_{1},t_{1}+\tau)$
is needed, where without altering $\mathbf{x}_{\text{s}}(t)$, \emph{i.e.},
$\dot{\mathbf{x}}_{\text{s}}=0$, we smoothly switch $\dot{\mathbf{x}}_{\text{s}}$
from $\dot{\mathbf{x}}_{\text{s}}(t_{1}^{-})$ to $\dot{\mathbf{x}}_{\text{s}}(t_{1}^{+})$
using a linear transition
\begin{equation}
\dot{\mathbf{x}}_{\text{t}}=\frac{\tau-t+t_{1}}{\tau}\dot{\mathbf{x}}_{\text{s}}(t_{1}^{-})+\frac{t-t_{1}}{\tau}\dot{\mathbf{x}}_{\text{s}}(t_{1}^{+}),\:t\in(t_{1},t_{1}+\tau)\label{eq:Transition-Phase-x-s-dynamic}
\end{equation}
which yields $\ddot{\mathbf{x}}_{\text{t}}=\frac{1}{\tau}(\dot{\mathbf{x}}_{\text{s}}(t_{1}^{+})-\dot{\mathbf{x}}_{\text{s}}(t_{1}^{-}))$.
As a result, the dynamical equation for $\mathbf{x}(t)$ becomes,
for $t\in(t_{1},t_{1}+\tau)$, 
\begin{align}
\dot{\mathbf{x}} & =\mathbf{R}(\dot{\mathbf{x}}_{\text{t}})\dot{\mathbf{x}}_{\text{r}}+\frac{1}{\tau}\nabla\mathbf{R}(\dot{\mathbf{x}}_{\text{t}})(\mathbf{I}_{3}\otimes(\dot{\mathbf{x}}_{\text{s}}(t_{1}^{+})-\dot{\mathbf{x}}_{\text{s}}(t_{1}^{-})))\mathbf{x}_{\text{r}}(t).\label{eq:Dynamical-Equation-Trajectory-x-Switch-x}
\end{align}

A sample implementation is summarized in Algorithm~\ref{alg:algorithm_equip}.

\subsubsection{Complexity Analysis}

\label{subsec:Complexity-Analysis}

As Algorithm~\ref{alg:algorithm_equip} is an online search scheme,
we investigate two different metrics for an understanding of its complexity:
the trajectory length of the \ac{uav} and the per-step computational
complexity every time the UAV adjusts its course.

The following proposition shows the trajectory length is linear in
the radius of the equipotential surface under certain conditions.
\begin{prop}[Upper bound of trajectory length]
\label{prop:upper-bound-trajectory-length}Suppose that $g_{k}(\mathbf{x})=b_{0}-10\log_{10}(d^{2}(\mathbf{x},\mathbf{u}_{k}))$.
For the balancing problem defined in (\ref{general_formula-balancing-problem})
with $f_{k}(g_{k}(\mathbf{x}),p_{k})=\log_{2}(1+p_{k}g_{k}(\mathbf{x}))$,
the trajectory length $L$ is upper bounded as
\begin{equation}
L\leq\pi(H_{0}+R)\sqrt{3\pi^{2}r_{1}^{2}+1}\label{eq:upper-bound-trajectory-length}
\end{equation}
where $H_{0}$ is the initial search altitude, $R$ is the radius
of the equipotential surface given in (\ref{eq:radius_2}), and $r_{1}$
is the measurement range.
\end{prop}
\begin{proof}
See Appendix~\ref{sec:Proof-of-Proposition-Trajectory-Length} in
\cite{ZheChe:J24}.
\end{proof}

Proposition~\ref{prop:upper-bound-trajectory-length} suggests that
reducing the initial search altitude decreases the upper bound of
the trajectory length $L$. Additionally, the trajectory length $L$
is related to the radius $R$ of the equipotential surface. In a special
case where ${\bf u}_{0}=[0,0,0]^{\text{T}}$, $K=1$, and $P_{0}>P_{\text{T}}>0$,
$R^{2}$ in (\ref{eq:radius_2}) simplifies to $R^{2}=P_{0}/(P_{0}-P_{\text{T}})\left(P_{0}/(P_{0}-P_{\text{T}})-1\right)d^{2}({\bf u}_{0},{\bf u}_{1})$,
which indicates that $R$ is linear in $d({\bf u}_{0},{\bf u}_{1})$.
Thus, the worst-case search trajectory length is linear in the BS-user
distance $d({\bf u}_{0},{\bf u}_{1})$.

Recall that $K$ is the number of users, $M$ is the number of measurements
used for local channel construction, and $N$ is the number of constraints
for resource allocation in problem~(\ref{general_formula-resource-allocation}).
The computational complexity of Steps~3 and 4 in Algorithm~\ref{alg:algorithm_equip}
is found as $O(KM+(K+N)^{3})$. See Appendix~\ref{sec:Per-step-Computational-Complexity}
in \cite{ZheChe:J24} for a detailed derivation of the computational
complexity.

\begin{algorithm}
Find a full-LOS initial position $\mathbf{x}_{0}\in\mathcal{S}\cap\tilde{\mathcal{D}}$.
Denote $\mathbf{x}_{\text{r}}(t)=r[\cos(\omega t),\,0,\,\sin(\omega t)]^{\text{T}}$.
\begin{enumerate}
\item \label{enu:Initialization}Initialization at $t=0$: $\tilde{\mathbf{x}}=\mathbf{x}_{0}$,
$\mathbf{x}_{\text{s}}(0)=\mathbf{x}_{0}$, and $\mathbf{x}(0)=\mathbf{x}_{\text{s}}(0)+\mathbf{x}_{\text{r}}(0)$.
\item \textbf{While} $x_{3}(t)\geq H_{\text{min}}$:
\item \textbf{Local channel construction:} Collect channel measurements
for each time slot $\triangle t$. Construct the local channel $\hat{g}_{k}(\mathbf{x})$
and obtain the estimation of $\nabla g_{k}(\mathbf{x})$ for each
user based on model (\ref{eq:linear-model-channel}) with parameters
(\ref{eq:least-square-solution}) estimated from the past $M$ \ac{los}
measurements.
\item Update the UAV location $\mathbf{x}(t)$ according to the following
cases:
\begin{enumerate}
\item \textbf{Phase~1 (LOS):} If $\mathbf{x}(t-\Delta t),\mathbf{x}(t)\in\tilde{\mathcal{D}}$,
\begin{enumerate}
\item If $f_{0}(g_{0}(\mathbf{x}_{\text{s}}(t)))>f_{0}(g_{0}(\tilde{\mathbf{x}}))$,
then update the optimal position $\tilde{\mathbf{x}}\leftarrow\mathbf{x}_{\text{s}}(t)$.
\item Compute the search direction $\mathbf{s}(t)=\mathcal{A}(\mathbf{x}_{\text{s}}(t);\mathbf{q}_{1}(\mathbf{x}_{\text{s}}(t)))+\mu_{\text{v}}\mathcal{V}(\mathbf{x}_{\text{s}}(t))$
from (\ref{eq:solution-to-local-approximation-equipotential-point}),
(\ref{eq:solution-dx-dt}), and (\ref{eq:closed-form-q-in-LOS}).
\item \label{enu:Phase1-(LOS)-update}Update $\mathbf{x}_{\text{s}}(t+\triangle t)=\mathbf{x}_{\text{s}}(t)+\mathbf{s}\triangle t$
and $\mathbf{x}(t+\triangle t)=\mathbf{x}(t)+(\mathbf{s}+\mathbf{R}(\mathbf{s})\dot{\mathbf{x}}_{\text{r}})\triangle t$
according to (\ref{eq:Dynamical-Equation-Trajectory-x}) and (\ref{eq:Dynamical-Equation-Trajectory-x-NoSwitch}).
\end{enumerate}
\item \textbf{Phase~2 (NLOS):} If $\mathbf{x}(t-\Delta t),\mathbf{x}(t)\neq\tilde{\mathcal{D}}$,
\begin{enumerate}
\item Compute the search direction $\mathbf{s}(t)=\mathcal{A}(\mathbf{x}_{\text{s}}(t);\mathbf{q}_{2}(\mathbf{x}_{\text{s}}(t)))+\mu_{\text{v}}\mathcal{V}(\mathbf{x}_{\text{s}}(t))$
from (\ref{eq:solution-to-local-approximation-equipotential-point}),
(\ref{eq:solution-dx-dt}), and (\ref{eq:closed-form-q2-in-NLOS}).
\item The same as Step \ref{enu:Phase1-(LOS)-update}.
\end{enumerate}
\item \textbf{Otherwise (Transition Phase):\label{enu:Transition-Phase}}
\begin{enumerate}
\item Let $t_{1}\leftarrow t-\triangle t$.
\item \label{enu:transition-phase-processing}Compute $\dot{\mathbf{x}}_{\text{t}}$
and $\bm{\delta}=\dot{\mathbf{x}}$ respectively according to (\ref{eq:Transition-Phase-x-s-dynamic})
and (\ref{eq:Dynamical-Equation-Trajectory-x-Switch-x}), by replacing
``$\dot{\mathbf{x}}_{\text{s}}(t_{1}^{-})$'' with $\mathbf{s}(t_{1})$
and replacing ``$\dot{\mathbf{x}}_{\text{s}}(t_{1}^{+})$'' with
$\mathbf{s}(t_{1}+\triangle t)$.
\item Update $\mathbf{x}_{\text{s}}(t+\triangle t)=\mathbf{x}_{\text{s}}(t)$
and $\mathbf{x}(t+\triangle t)=\mathbf{x}(t)+\bm{\delta}\triangle t$.
\item Repeat from Step \ref{enu:transition-phase-processing} until $t\geq t_{1}+\tau$.
\end{enumerate}
\end{enumerate}
\item \textbf{End while}; output $\tilde{\mathbf{x}}$ as the best position
found on $\mathcal{S}\cap\tilde{\mathcal{D}}$.
\end{enumerate}
\caption{Superposed LOS discovery and measurement collection trajectory with
unknown user locations}

\label{alg:algorithm_equip}
\end{algorithm}

\section{Numerical Results}

\label{sec:Numerical-Results}

In this section, we present the experimental findings conducted on
two real 3D urban maps.

\subsection{Environment Setup and Scenarios}

\label{subsec:Environment-Setup-and-Scenarios}

Two city maps of different areas in Beijing, China are used to evaluate
the proposed scheme. As shown in Fig.~\ref{fig:city_maps}, map~A
represents a sparse commercial area with the building coverage ratio
(BCR) and floor area ratio (FAR) \cite{GonCreHerRod:J13} as $18\%$
and $1.0$, respectively, while map~B represents a dense residential
area with the BCR and FAR as $33\%$ and $1.86$, respectively. The
minimum flying altitudes $H_{\text{min}}$ are set as $29$ and $62$
meters on maps A and B, respectively, which correspond to the minimum
height of the top $20\%$ tallest buildings. The users are placed
uniformly at random in the non-building area for $2000$ repetitions,
with the \ac{bs} placed above an arbitrarily chosen building. For
the proposed scheme, $M$, $\omega$, $r$, $\mu_{\text{v}}$, and
$\tau$ are set as $100$, $\pi/25$, $25$, $1$, and $5$, respectively.

\selectlanguage{english}%
Two application scenarios are evaluated in our experiments. For a
sum-rate application, one UAV is placed to establish \ac{los} relay
channels for $K$ ground users and a \ac{bs} under decode-and-forward
relaying. Consider the path loss model of millimeter wave cellular
reported in \cite{CheMitGes:T21,HerHolStaBli:J10}, where \foreignlanguage{american}{the
deterministic \ac{los} channel gain $g_{k}(\mathbf{x})$ is modeled
as $g_{k}(\mathbf{x})=46.53+20.0\log_{10}d(\mathbf{x},\mathbf{u}_{k})$.
The variance $\sigma^{2}$ of measurement uncertainty $\xi$ in (\ref{eq:measurement_with_noise})
is set as $5$~dB. It is assumed that the \ac{mmw} beam alignment
has been done for every \ac{uav} position. For a balancing application,
}one \ac{uav} is placed to \foreignlanguage{american}{provide location
estimation services for $8$ sensing targets and maintain a backhaul
communication link with a \ac{bs} simultaneously. Particularly, the
}objective functions for sensing are specified as \ac{snr}, {\em i.e.},
$f_{k}(g_{k}(\mathbf{x}),p_{k})=p_{k}g_{k}(\mathbf{x})/(\bar{N}_{0}W)$
where the \foreignlanguage{american}{noise power spectral density
$\bar{N}_{0}$ is set as $-150$~dBm/Hz, and the bandwidth $W$ is
set as $1$~GHz.}
\begin{figure}
\selectlanguage{american}%
\begin{centering}
\includegraphics[width=1\columnwidth]{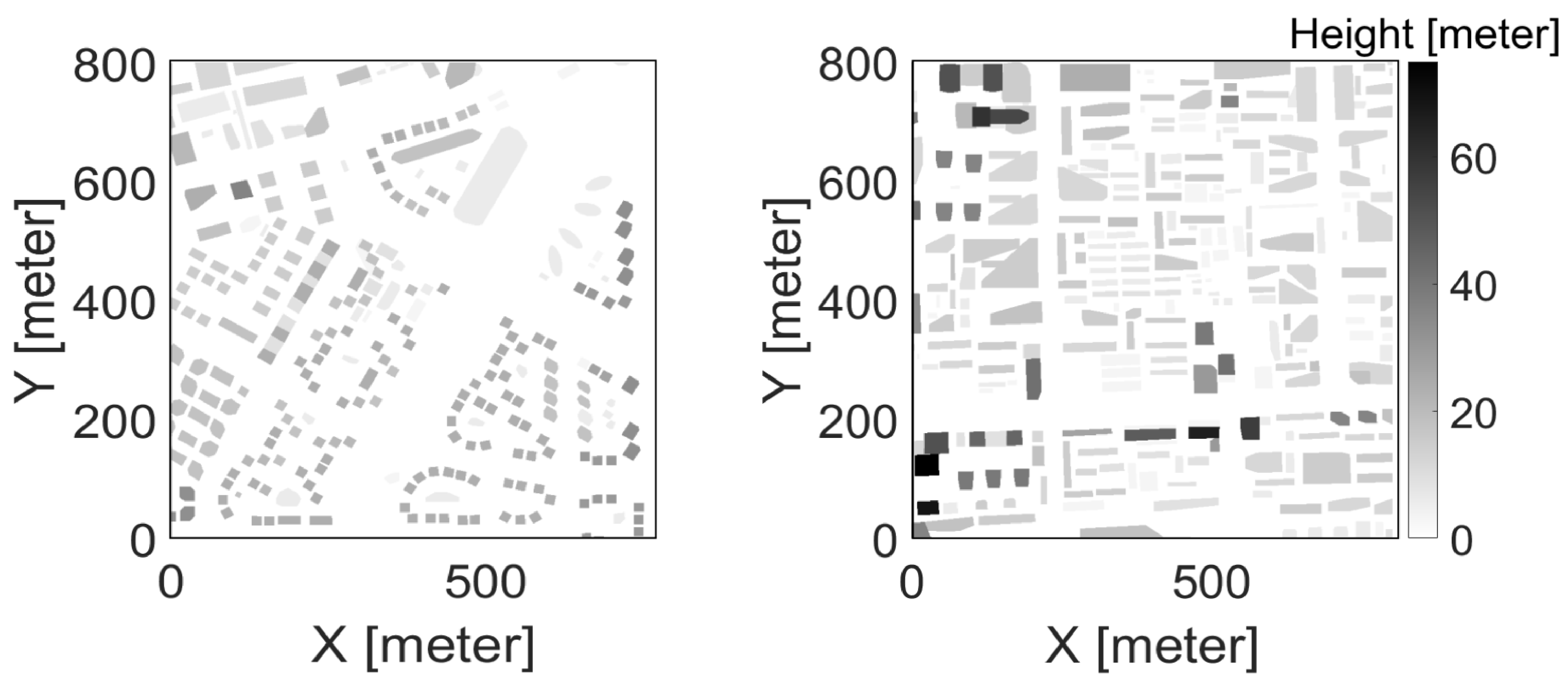}
\par\end{centering}
\caption{Map A (left) is a sparse commercial area, and map B (right) is a dense
residential area in Beijing, China.}

\label{fig:city_maps}\selectlanguage{english}%
\end{figure}

\selectlanguage{american}%
The baseline schemes are listed below, and the exhaustive search schemes
are implemented with a $5$-meter step size.
\begin{itemize}
\item \textit{Exhaustive 3D search (Exh3D)}: This scheme performs an exhaustive
search in 3D space above the area of interest.
\item \textit{Exhaustive 2D searc}\emph{h (Exh2D) }\cite{LyuZenZhaLim:J17}:
This scheme performs an exhaustive search at a height of $H_{\text{min}}+50$.
\item \textit{Statistical geometry (Statis)} \cite{AlhKanLar:J14,CheHua:J22}:
The average channel gain from the \ac{uav} position $\mathbf{x}$
to the $k$th user is formulated as $\tilde{g}_{k}(\mathbf{x})=\text{P}_{\text{L}}(\mathbf{x})g_{k}(\mathbf{x})+(1-\text{P}_{\text{L}}(\mathbf{x}))(g_{k}(\mathbf{x})+\phi(\mathbf{x}))$
where the power penalty $\phi(\mathbf{x})$ for \ac{nlos} link is
set as $-30$ dB, and $\text{P}_{\text{L}}(\mathbf{x})$ is the \ac{los}
probability at $\mathbf{x}$, which is defined as $\text{P}_{\text{L}}(\mathbf{x})=1/(1+a_{\text{e}}\times\text{exp}(-b_{\text{e}}(\text{arctan}(x_{3}/\sqrt{\|\mathbf{x}-\mathbf{u}_{k}\|_{2}^{2}-x_{3}^{2}})-a_{\text{e}})))$
where the parameters $a_{\text{e}}$ and $b_{\text{e}}$ are learned
from maps.
\item \textit{Relaxed analytical geometry (RAG)} \cite{YiZhuZhuXia:J22}:
This scheme models the city structure using polyhedrons and determines
the LOS conditions via a set of constraints obtained from analytical
geometry. The UAV position optimization problem is solved by Lagrangian
relaxation and successive convex approximation (SCA).
\end{itemize}

Note that both the statistical geometry scheme and the relaxed analytical
geometry scheme require user locations and channel model parameters,
and hence, they serve for performance benchmarking only. Additionally,
we test a genius-aided version of the proposed scheme with known user
locations and channel models, and thus, the measurement range $r_{1}$
is set to $0$. This is to set the benchmark for the best possible
performance of searching on the equipotential surface.

\subsection{UAV-assisted Communication and Sensing}

\begin{figure}
\begin{centering}
\includegraphics[width=0.93\columnwidth]{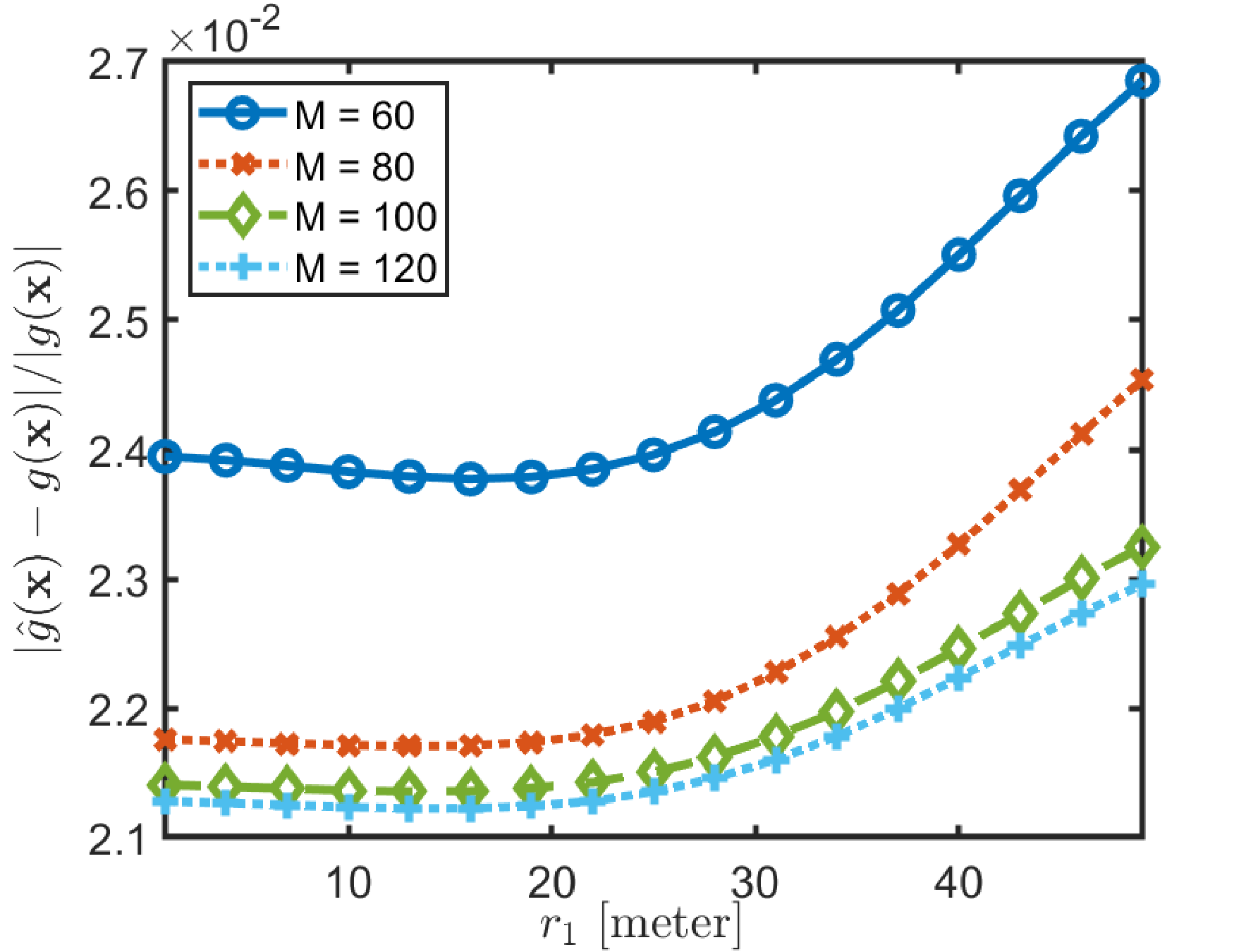}
\par\end{centering}
\caption{\label{fig:channel_gain_error}Normalized estimation error of channel
gain $g(\mathbf{x})$ versus measurement range $r_{1}$.}
\end{figure}
\begin{figure}
\begin{centering}
\includegraphics[width=0.95\columnwidth]{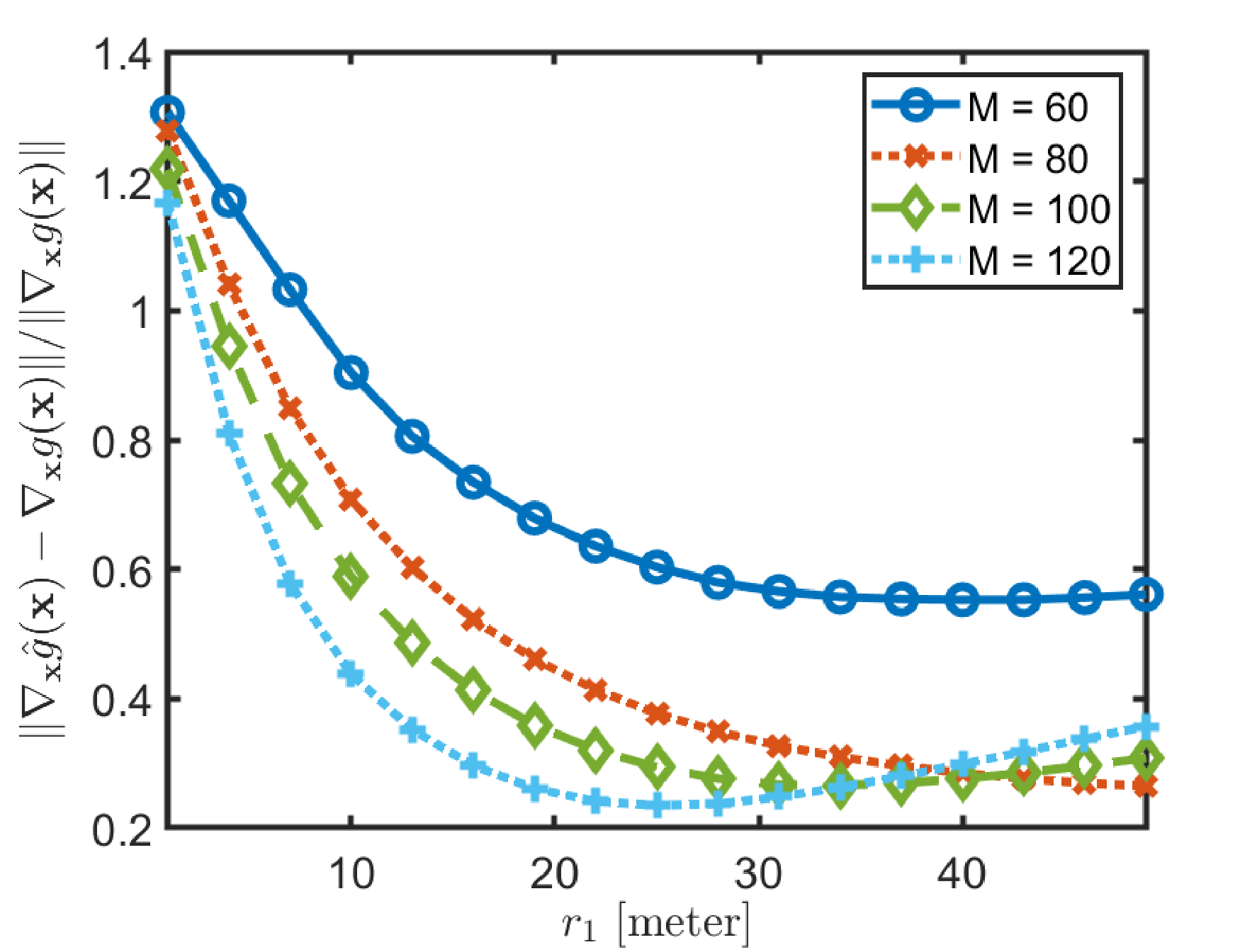}
\par\end{centering}
\caption{\label{fig:channel_gradient_error}Normalized estimation error of
local gradient of channel gain, {\em i.e.}, $\nabla g(\mathbf{x})$,
versus measurement range $r_{1}$.}
\end{figure}

Fig.~\ref{fig:channel_gain_error} shows the normalized estimation
error of channel gain $g(\mathbf{x})$ under different settings of
$M$ and $r_{1}$. It is found that the error in estimating $g(\mathbf{x})$
exhibits a monotonic decrease \ac{wrt} the increase of $M$ while
it is not monotonic \ac{wrt} $r_{1}$. This observation suggests
the existence of an optimal value for $r_{1}$, {\em e.g.}, $r_{1}\approx18$
when $M=60$, which confirms the analytical results in Theorem~\ref{thm:mse-of-estimated-channel-gain}.
Fig.~\ref{fig:channel_gradient_error} demonstrates the normalized
estimation error of the gradient of channel gain $\nabla g(\mathbf{x})$.
It is observed that a small measurement range $r_{1}$ leads to a
large estimation error on the local gradient of $g(\mathbf{x})$ as
the measurement is not geographically diverse. Yet, when $r_{1}$
is too large, the first-order local \ac{los} model $\hat{g}(\mathbf{x})$
becomes less accurate, leading to a slight increase of the estimation
error.

\begin{table}
\caption{\label{tab:comparison-of-capacity-sum-rate}Capacity on two maps for
a sum-rate application {[}${\rm Gbps}${]}}

\renewcommand{\arraystretch}{1.2}
\centering{}%
\begin{tabular}{>{\raggedright}m{0.1\columnwidth}>{\raggedright}m{0.07\columnwidth}>{\raggedright}m{0.08\columnwidth}>{\raggedright}m{0.05\columnwidth}>{\centering}m{0.1\columnwidth}>{\centering}m{0.18\columnwidth}>{\centering}m{0.07\columnwidth}}
\hline 
\centering{} & \centering{}\textbf{Statis} & \centering{}\textbf{Exh2D} & \centering{}\textbf{RAG} & \centering{}\textbf{Proposed} & \centering{}\textbf{Genius-aided} & \centering{}\textbf{Exh3D}\tabularnewline
\hline 
\centering{}Map A & \centering{}3.21 & \centering{}3.52 & \centering{}3.72 & \centering{}3.82 & 3.84 & 3.85\tabularnewline
\hline 
\centering{}Map B & \centering{}1.70 & \centering{}1.84 & \centering{}2.08 & \centering{}2.17 & 2.24 & 2.30\tabularnewline
\hline 
\end{tabular}
\end{table}

\begin{figure}
\begin{centering}
\includegraphics[width=0.93\columnwidth]{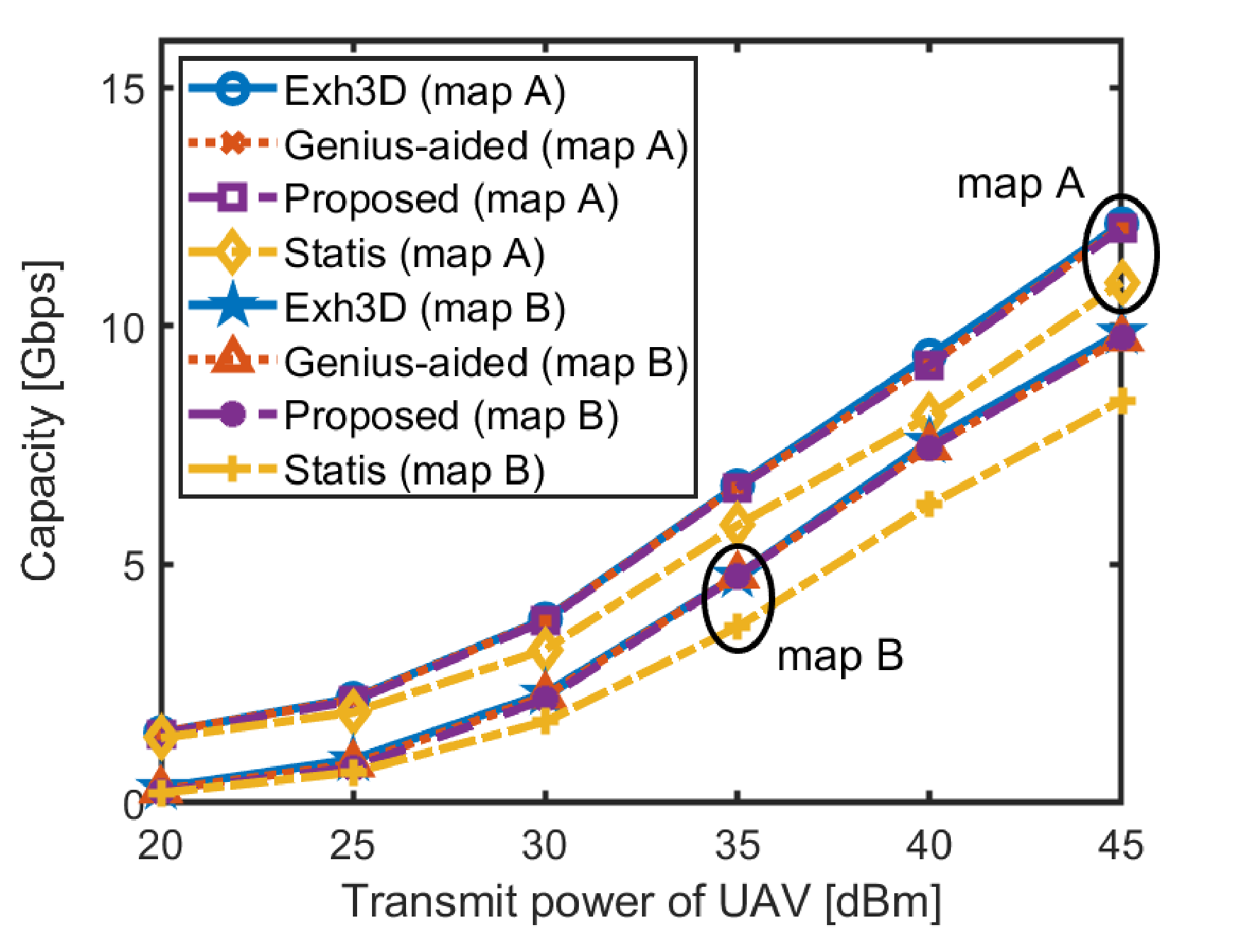}
\par\end{centering}
\caption{\label{fig:Capacity-transmit-power}System capacity versus transmit
power of UAV in a sum-rate application}
\end{figure}

Table~\ref{tab:comparison-of-capacity-sum-rate} summarizes the system
capacity, {\em i.e.}, $\min\{f_{0}(g_{0}(\mathbf{x})),\sum_{k\in\mathcal{K}}\{f_{k}(g_{k}(\mathbf{x}),p_{k}\text{)}\}\}$,
of different schemes on two maps. The proposed scheme achieves over
$94\%$ to that of the exhaustive 3D search, only 3\% below the genius-aided
scheme which requires perfect user locations and channel parameters.
Such a result suggests that the globally optimal solution likely resides
on the equipotential surface, and the proposed scheme can search near
the locally reconstructed equipotential surface without significant
deviation. Although the relaxed analytical geometry scheme achieves
over $90\%$ to the exhaustive 3D search over two maps, it requires
complete knowledge of city topology and additional computational cost
for polygonal approximation of buildings. The suboptimal performance
of the exhaustive 2D search is primarily due to its inability to leverage
height flexibility. When compared to other schemes with \ac{los}
guarantee, the poor performance of the statistical geometry scheme
emerges from the inherent uncertainty of the LOS status of potentially
best UAV positions found through this scheme.

Fig.~\ref{fig:Capacity-transmit-power} demonstrates the system capacity
under varying UAV transmit powers, with the transmit power of \ac{bs}
set to $K$ times that of the UAV. It is observed that the proposed
scheme closely approaches the performance of the exhaustive 3D search
regardless of transmit power. Although a denser map makes it harder
to find the globally optimal solution, the proposed scheme still achieves
over $94\%$ to the exhaustive 3D search on map~B while it achieves
over $99\%$ to the exhaustive 3D search on map~A under $P_{\text{T}}=30$~dBm.

Fig.~\ref{fig:Capacity-user-number} shows the system capacity under
different user numbers in the sum-rate application. It is found that
the proposed scheme achieves over $98\%$ to the exhaustive 3D search
on map~A under all tested user numbers, while it achieves up $90\%$
on map~B, which represents a denser urban area. The reason is that
denser urban topology and more users result in much more complicated
\ac{los} distributions. Thus, gathering a sufficient number of \ac{los}
measurements in a local area for channel estimation becomes more challenging,
resulting in increased construction error of the equipotential surface.
\begin{figure}
\begin{centering}
\includegraphics[width=0.93\columnwidth]{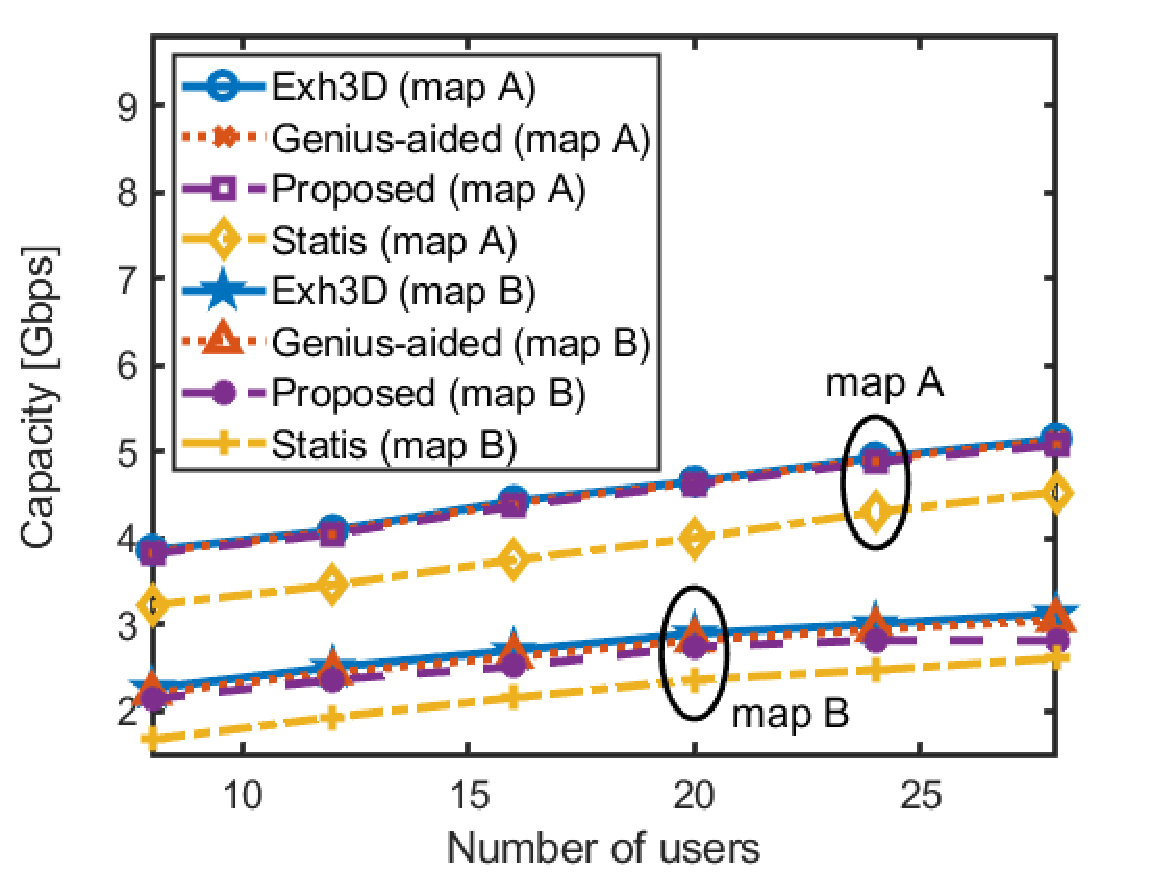}
\par\end{centering}
\caption{\label{fig:Capacity-user-number}System capacity versus number of
users in a sum-rate application}
\end{figure}
\begin{figure}
\begin{centering}
\includegraphics[width=0.93\columnwidth]{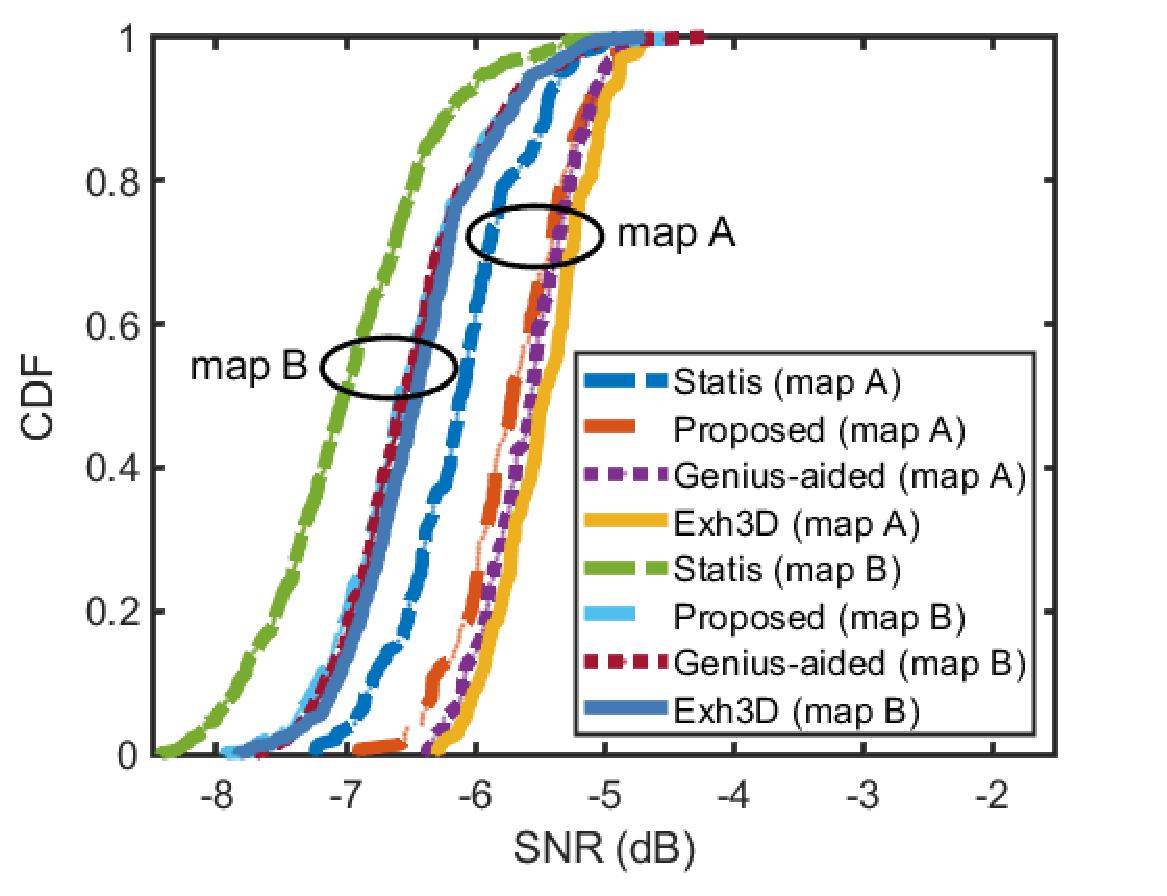}
\par\end{centering}
\caption{\label{fig:CDF_versus_SNR}CDF versus $\text{SNR}$ of the weakest
sensing link in a balancing application.}
\end{figure}

In Fig.~\ref{fig:CDF_versus_SNR}, we examine the proposed scheme
in a balancing application. As observed, the proposed scheme substantially
increases the \ac{snr} of the weakest sensing link over the statistical
geometry scheme. In particular, it is observed that the \ac{cdf}
curve of the proposed scheme closely coincides with that of the exhaustive
3D search on map~B, numerically confirming that the proposed scheme
can find the near-optimal solution in complex urban environments.

Fig.~\ref{fig:Convergence-process} shows two examples of convergence
process of the proposed scheme. The velocity of the \ac{uav} is set
as $10$~meters per second. It is observed that the system capacity
increases with the search time because the proposed scheme searches
in the direction of increasing the objective value in the full-LOS
region or maintaining the current objective value in the non-full-LOS
region.

Define the convergence time as the search time at which the system
performance no longer improves, in which case, we consider that the
best \ac{uav} position is discovered. Fig.~\ref{fig:Convergence-time}
summarizes the convergence time versus number of users. It is observed
that the convergence time decreases with the increase of number of
users $K$. The reason is that the radius of the equipotential surface
decreases with the increase of $K$, which is analytically shown in
Proposition~\ref{prop:shape_of_equipotential_surface}. Since the
maximum height of the equipotential surface is smaller, the trajectory
length will be shorter as it takes less time to reach the minimum
flight height.

Table~\ref{tab:comparison-of-computational-time} summarizes the
average trajectory lengths required by the four online schemes on
two maps. Notably, the genius-aided scheme with known user locations
and channel models requires only a few hundred meters to approximate
a near-optimal solution, confirming the efficiency of the \ac{los}
discovery strategy proposed in Section~\ref{subsec:LOS-aware-Search-Trajectory}.
The proposed scheme demands a longer trajectory of approximately $3$~kilometers
due to its reliance on spiral trajectories for data collection during
the search process. Despite this, the proposed scheme still demonstrates
considerable efficacy, significantly reducing search complexity compared
to the exhaustive search schemes.

\begin{figure}
\begin{centering}
\includegraphics[width=0.93\columnwidth]{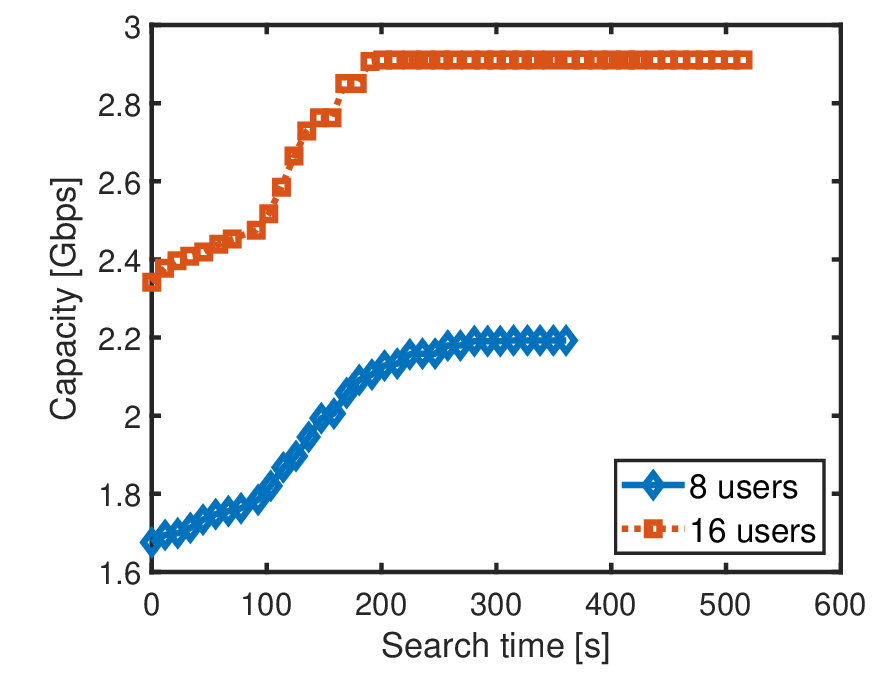}
\par\end{centering}
\caption{\label{fig:Convergence-process}Convergence process: system capacity
versus search time.}
\end{figure}

\begin{figure}
\begin{centering}
\includegraphics[width=0.93\columnwidth]{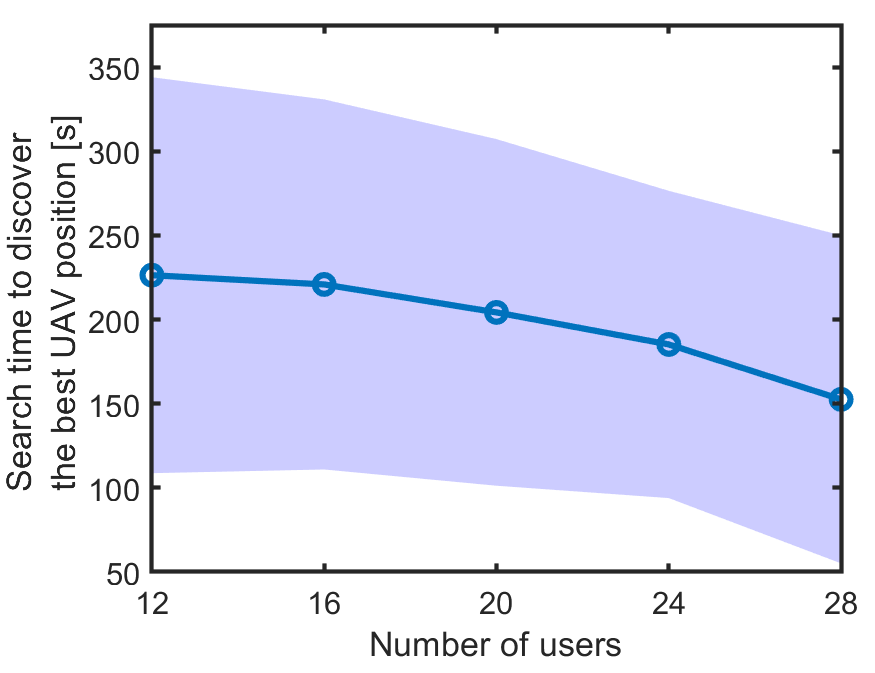}
\par\end{centering}
\caption{\label{fig:Convergence-time}The search time to discover the best
\ac{uav} position versus the number of users. The shaded area represents
a possible search time realization during our $2000$ simulations.}
\end{figure}

\begin{table}
\caption{\label{tab:comparison-of-computational-time}Comparison of Average
Trajectory Length {[}kilometer{]}}

\renewcommand{\arraystretch}{1.2}
\centering{}%
\begin{tabular}{>{\raggedright}m{0.1\columnwidth}>{\raggedright}m{0.2\columnwidth}>{\centering}m{0.14\columnwidth}>{\centering}m{0.14\columnwidth}>{\centering}m{0.14\columnwidth}}
\hline 
\centering{} & \centering{}\textbf{Genius-aided} & \centering{}\textbf{Proposed} & \centering{}\textbf{Exh2D} & \centering{}\textbf{Exh3D}\tabularnewline
\hline 
\centering{}Map A & \centering{}0.2604 & 3.272 & 1920 & 42240\tabularnewline
\hline 
\centering{}Map B & \centering{}0.2428 & 3.051 & 1920 & 42240\tabularnewline
\hline 
\end{tabular}
\end{table}

\section{Conclusion}

\label{sec:Conclusion}

This paper developed an efficient online trajectory for optimal \ac{uav}
placement without prior knowledge of user locations, channel model
parameters, and terrain structure. We analytically characterized the
equipotential surface and proposed an \ac{los} discovery trajectory
on it, utilizing perturbation theory to guide the \ac{uav} search
direction. Additionally, we developed a class of spiral trajectories
to construct a local channel map in the \ac{los} regime using local
polynomial regression without depending on user locations. After deriving
the optimal measurement pattern, we minimized the \ac{mse} of the
locally estimated channel gain and determined the optimal measurement
range. Experimental results on real urban maps demonstrated that our
approach achieves over $94\%$ of the performance of a 3D exhaustive
search scheme with just a $3$-kilometer search in a complex environment.


\appendices{}


\section{Proof of Proposition~\ref{prop:existence-condition-special-balancing-problem}}

\label{sec:Proof-of-Proposition-existence-condition}

First, the optimal power allocation in the balancing problem is expressed
as
\begin{equation}
p_{k}^{*}=\frac{P_{\text{T}}}{g_{k}(\mathbf{x})\sum_{k'\in\mathcal{K}}(1/g_{k'}(\mathbf{x}))}.\label{eq:optimal-power-allocation-balancing}
\end{equation}
which is achieved when all of the users share the same received \ac{snr},
{\em i.e.}, $p_{1}g_{1}(\mathbf{x})=p_{2}g_{2}(\mathbf{x})=\cdots=p_{K}g_{K}(\mathbf{x})$.
Otherwise, the \ac{uav} can enhance the worst link performance by
allocating more power to that link. Since there is a maximum power
constraint $\sum_{k=1}^{K}p_{k}\leq P_{\text{T}}$, the optimal solution
$p_{k}^{*}$ in (\ref{eq:optimal-power-allocation-balancing}) can
be derived. Second, the condition in (\ref{eq:sufficient_condition_existence-balancing-problem})
can be equivalently rewritten as $F(\bm{g}(\mathbf{x}_{0}^{\text{m}}))F(\bm{g}(\mathbf{x}_{\text{\text{u}}}^{\text{m}}))\leq0$.
Third, according to the Bolzano-Cauchy theorem and the continuity
of $F(\bm{g}(\mathbf{x}))$, there exists an equipotential point $\mathbf{x}$
satisfying $F(\bm{g}(\mathbf{x}))=0$, confirming the existence of
the equipotential surface.

\section{Proof of Proposition~\ref{prop:shape_of_equipotential_surface}}

\label{sec:Proof-of-Proposition-sphere}

Recall that the optimal power allocation in the balancing problem
is derived in (\ref{eq:optimal-power-allocation-balancing}). Thus,
the objective function of the UAV-user link is rewritten as
\begin{equation}
F_{\text{u}}(\bm{g}_{\text{u}}(\mathbf{x}),\mathbf{p}^{*})=\log_{2}(1+P_{\text{T}}/{\textstyle \sum_{k'\in\mathcal{K}}}(1/g_{k'}(\mathbf{x}))).
\end{equation}
Given $g_{k}(\mathbf{x})=b_{0}-10\log_{10}(d^{2}(\mathbf{x},\mathbf{u}_{k}))$,
the equation $F(\bm{g}(\mathbf{x}))=0$ is simplified as $P_{\text{T}}d^{2}(\mathbf{x},\mathbf{u}_{0})-{\textstyle P_{0}\sum_{k'\in\mathcal{K}}d^{2}(\mathbf{x},\mathbf{u}_{k})}=0$
which can be rewritten as a spherical equation with specified center
and radius in (\ref{eq:radius_2}).

\section{Proof of Theorem~\ref{thm:variance-minimization}}

\label{sec:Proof-of-Theorem-variance-minimization}

We first derive the following lemma on trace of the inverse of a positive
definite and symmetric matrix.
\begin{lem}[Trace of the inverse of a positive definite and symmetric matrix]
\label{lem:trace-inverse-positive-definite}Let $\mathbf{A}\in\mathcal{M}_{N}$
be a positive definite symmetric matrix, and $a_{ii}$ be the $i$th
diagonal element. Then,
\[
\text{{\rm tr}}\{\mathbf{A}^{-1}\}\geq\sum_{i=1}^{N}a_{ii}^{-1}
\]
with the equality achieved for $a_{ij}=0$, $\forall i\neq j$.
\end{lem}
\begin{proof}
See Appendix~\ref{sec:Supplementary-Materials-2} in \cite{ZheChe:J24}.
\end{proof}

According to \cite{FanGij:B96,SunChe:J22}, the variance of $\boldsymbol{\theta}^{\text{(e)}}$
is given by $\text{tr}\{\sigma^{2}(\tilde{\boldsymbol{\mathbf{X}}}^{\text{T}}\tilde{\boldsymbol{\mathbf{X}}})^{-1}\}$.
Using the results in Lemma~\ref{lem:trace-inverse-positive-definite},
one can minimize $\text{tr}\{(\tilde{\boldsymbol{\mathbf{X}}}^{\text{T}}\tilde{\boldsymbol{\mathbf{X}}})^{-1}\}$
with the optimal conditions on $\tilde{\boldsymbol{\mathbf{X}}}$.

Since $\tilde{\boldsymbol{\mathbf{X}}}^{\text{T}}\tilde{\boldsymbol{\mathbf{X}}}$
is invertible, none of its eigenvalues is equal to $0$. Given the
definition of $\tilde{\mathbf{X}}$, all the primary minors of $\tilde{\boldsymbol{\mathbf{X}}}^{\text{T}}\tilde{\boldsymbol{\mathbf{X}}}$
are no less than $0$. Hence, $\tilde{\boldsymbol{\mathbf{X}}}^{\text{T}}\tilde{\boldsymbol{\mathbf{X}}}$
is a positive definite symmetric matrix. According to Lemma~\ref{lem:trace-inverse-positive-definite},
\begin{align*}
\text{tr}\{(\tilde{\boldsymbol{\mathbf{X}}}^{\text{T}}\tilde{\boldsymbol{\mathbf{X}}})^{-1}\} & \geq\frac{1}{M}+\sum_{j=1}^{3}\frac{1}{\sum_{m=1}^{M}(x_{mj}-c_{0j})^{2}}
\end{align*}
with the minimum achieved when $\sum_{m=1}^{M}(x_{mj}-c_{0j})=0$,
for $\forall j\in\{1,2,3\}$, and $\sum_{m=1}^{M}(x_{mj}-c_{0j})(x_{mj'}-c_{0j'})=0$,
for $\forall j,j'\in\{1,2,3\}$ and $j\neq j'$.

Using the Arithmetic Geometric Mean Inequality (AGMI) and the fact
that $\sum_{j=1}^{3}\sum_{m=1}^{M}(x_{mj}-c_{0j})^{2}\leq Mr_{1}^{2}$,
we have

\begin{align*}
\text{tr}\{(\tilde{\boldsymbol{\mathbf{X}}}^{\text{T}}\tilde{\boldsymbol{\mathbf{X}}})^{-1}\} & \text{\ensuremath{\geq}}\frac{1}{M}+\frac{3}{\sqrt[3]{\prod_{j=1}^{3}\sum_{m=1}^{M}(x_{mj}-c_{0j})^{2}}}\\
 & =\frac{1}{M}+\frac{9}{Mr_{1}^{2}}
\end{align*}
with the minimum achieved when $\sum_{m=1}^{M}(x_{mj}-c_{0j})^{2}=Mr_{1}^{2}/3$
for $j\in\{1,2,3\}$. Thus, the lower bound of $\text{tr}\{\sigma^{2}(\tilde{\boldsymbol{\mathbf{X}}}^{\text{T}}\tilde{\boldsymbol{\mathbf{X}}})^{-1}\}$
is correspondingly obtained.

\section{Proof of Theorem~\ref{thm:mse-of-estimated-channel-gain}}

\label{sec:Proof-of-Theorem-MSE-Channel-Gain}

The \ac{mse} of the estimated channel gain at $\mathbf{x}$ is given
by 
\begin{align}
\mathbb{E}\left\{ \left(\hat{g}(\mathbf{x})-g(\mathbf{x})\right)^{2}\right\}  & =\mathbb{V}\left\{ \hat{g}(\mathbf{x})\right\} +\left(\mathbb{E}\{\hat{g}(\mathbf{x})-g(\mathbf{x})\}\right)^{2}\label{eq:mse-g}
\end{align}

Next, the bounds of $\mathbb{V}\{\hat{g}(\mathbf{x})\}$ and $|\mathbb{E}\{\hat{g}(\mathbf{x})-g(\mathbf{x})\}|$
with $r_{0}=d(\mathbf{x},\mathbf{c}_{0})$ are derived as follows.

1) Derivation of $\mathbb{V}\{\hat{g}(\mathbf{x})\}$: According to
the proof of Theorem~\ref{thm:variance-minimization}, the variances
of $\hat{\alpha}$ and $\hat{\boldsymbol{\beta}}$ are given by $\mathbb{V}\{\hat{\alpha}\}=\sigma^{2}/M$,
and $\mathbb{V}\{\hat{\beta}_{j}\}=3\sigma^{2}/Mr_{1}^{2}$. Thus,
$\mathbb{V}\{\hat{g}(\mathbf{x})\}$ is given by
\begin{align}
\mathbb{V}\left\{ \hat{g}(\mathbf{x})\right\}  & =\mathbb{V}\left\{ \hat{\alpha}+\hat{\boldsymbol{\beta}}^{\text{T}}(\mathbf{x}-\mathbf{c}_{0})\right\} =\frac{\sigma^{2}}{M}+\frac{3\sigma^{2}r_{0}^{2}}{Mr_{1}^{2}}.\label{eq:variance-hat-gx}
\end{align}

2) Derivation of $|\mathbb{E}\{\hat{g}(\mathbf{x})-g(\mathbf{x})\}|$:
Since $\{\mathbf{x}_{m}\}$ satisfies conditions (i)-(iii) in Theorem~\ref{thm:variance-minimization},
$(\tilde{\boldsymbol{\mathbf{X}}}^{\text{T}}\tilde{\boldsymbol{\mathbf{X}}})^{-1}$
is simplified as
\[
(\tilde{\boldsymbol{\mathbf{X}}}^{\text{T}}\tilde{\boldsymbol{\mathbf{X}}})^{-1}=\frac{1}{M}\left[\begin{array}{cc}
1 & \mathbf{0}\\
\mathbf{0} & \frac{3}{r_{1}^{2}}\mathbf{I}
\end{array}\right]
\]
where $\mathbf{I}$ is a $3\times3$ unit matrix. Therefore, $\left(\tilde{\boldsymbol{\mathbf{X}}}^{\text{T}}\tilde{\boldsymbol{\mathbf{X}}}\right)^{-1}\tilde{\boldsymbol{\mathbf{X}}}^{\text{T}}$
is derived as
\[
\left(\tilde{\boldsymbol{\mathbf{X}}}^{\text{T}}\tilde{\boldsymbol{\mathbf{X}}}\right)^{-1}\tilde{\boldsymbol{\mathbf{X}}}^{\text{T}}=\frac{1}{M}\left[\begin{array}{ccc}
1 & \cdots & 1\\
\frac{3}{r_{1}^{2}}(\mathbf{x}_{1}-\mathbf{c}_{0}) & \cdots & \frac{3}{r_{1}^{2}}(\mathbf{x}_{M}-\mathbf{c}_{0})
\end{array}\right].
\]

Recall that $g(\mathbf{x})$ satisfies that Lipschitz condition in
(\ref{eq:Lipschitz-g-upper}) and (\ref{eq:Lipschitz-g-lower}). Hence,
we have
\begin{align*}
y_{m} & =g(\mathbf{x}_{m})+\xi_{m}\\
 & =\left[\begin{array}{cc}
1 & (\mathbf{x}_{m}-\mathbf{c}_{0})^{\text{T}}\end{array}\right]\left[\begin{array}{c}
g(\mathbf{c}_{0})\\
\nabla g(\mathbf{c}_{0})
\end{array}\right]+e(\mathbf{x}_{m},\mathbf{c}_{0})+\xi_{m}
\end{align*}
where $|e(\mathbf{x}_{m},\mathbf{c}_{0})|\leq L_{g}d^{2}(\mathbf{x}_{m},\mathbf{c}_{0})/2\leq L_{g}r_{1}^{2}/2$,
and $\xi_{m}\sim N(0,\sigma^{2})$ is the measurement noise at $\mathbf{x}_{m}$.

Recall that $\mathbf{y}=[y_{1},y_{2},\dots,y_{M}]^{\text{T}}$. Then,
we have
\begin{align*}
\mathbf{y} & =\left[\begin{array}{cc}
1 & (\mathbf{x}_{1}-\mathbf{c}_{0})^{\text{T}}\\
1 & (\mathbf{x}_{2}-\mathbf{c}_{0})^{\text{T}}\\
\vdots & \vdots\\
1 & (\mathbf{x}_{M}-\mathbf{c}_{0})^{\text{T}}
\end{array}\right]\left[\begin{array}{c}
g(\mathbf{c}_{0})\\
\nabla g(\mathbf{c}_{0})
\end{array}\right]+\left[\begin{array}{c}
e(\mathbf{x}_{1},\mathbf{c}_{0})+\xi_{1}\\
e(\mathbf{x}_{2},\mathbf{c}_{0})+\xi_{2}\\
\vdots\\
e(\mathbf{x}_{M},\mathbf{c}_{0})+\xi_{M}
\end{array}\right]\\
 & =\tilde{\boldsymbol{\mathbf{X}}}\left[\begin{array}{c}
g(\mathbf{c}_{0})\\
\nabla g(\mathbf{c}_{0})
\end{array}\right]+\mathbf{e}+\boldsymbol{\xi}
\end{align*}
where $\mathbf{e}=[e(\mathbf{x}_{1},\mathbf{c}_{0}),e(\mathbf{x}_{2},\mathbf{c}_{0}),\dots,e(\mathbf{x}_{M},\mathbf{c}_{0})]^{\text{T}}$,
and $\bm{\xi}=[\xi_{1},\xi_{2},\dots,\xi_{M}]^{\text{T}}$ satisfying
$\mathbb{E}\{\bm{\xi}\}=\mathbf{0}$.

Denote $\bar{\bm{\theta}}=[g(\mathbf{c}_{0}),\nabla g(\mathbf{c}_{0})^{\text{T}}]^{\text{T}}$.
The expectation of $\hat{\bm{\theta}}-\bar{\bm{\theta}}$ is given
by
\begin{align*}
\mathbb{E}\{\hat{\bm{\theta}}-\bar{\bm{\theta}}\} & =\mathbb{E}\left\{ \left(\tilde{\boldsymbol{\mathbf{X}}}^{\text{T}}\tilde{\boldsymbol{\mathbf{X}}}\right)^{-1}\tilde{\boldsymbol{\mathbf{X}}}^{\text{T}}\mathbf{y}\right\} -\bar{\bm{\theta}}\\
 & =\mathbb{E}\left\{ \left(\tilde{\boldsymbol{\mathbf{X}}}^{\text{T}}\tilde{\boldsymbol{\mathbf{X}}}\right)^{-1}\tilde{\boldsymbol{\mathbf{X}}}^{\text{T}}\tilde{\boldsymbol{\mathbf{X}}}\bar{\bm{\theta}}+\left(\tilde{\boldsymbol{\mathbf{X}}}^{\text{T}}\tilde{\boldsymbol{\mathbf{X}}}\right)^{-1}\tilde{\boldsymbol{\mathbf{X}}}^{\text{T}}(\mathbf{e}+\boldsymbol{\xi})\right\} \\
 & -\bar{\bm{\theta}}\\
 & =\left(\tilde{\boldsymbol{\mathbf{X}}}^{\text{T}}\tilde{\boldsymbol{\mathbf{X}}}\right)^{-1}\tilde{\boldsymbol{\mathbf{X}}}^{\text{T}}\mathbf{e}\\
 & =\frac{1}{M}\left[\begin{array}{c}
\sum_{m=1}^{M}e(\mathbf{x}_{m},\mathbf{c}_{0})\\
\frac{3}{r_{1}^{2}}\sum_{m=1}^{M}(\mathbf{x}_{m}-\mathbf{c}_{0})\cdot e(\mathbf{x}_{m},\mathbf{c}_{0})
\end{array}\right].
\end{align*}

Since $\hat{g}(\mathbf{x})=\hat{\alpha}+\hat{\boldsymbol{\beta}}^{\text{T}}(\mathbf{x}-\mathbf{c}_{0})$
and $g(\mathbf{x})=g(\mathbf{c}_{0})+\nabla g(\mathbf{c}_{0})^{\text{T}}(\mathbf{x}-\mathbf{c}_{0})+e(\mathbf{x},\mathbf{c}_{0})$,
we have
\begin{align*}
\mathbb{E}\{\hat{g}(\mathbf{x})-g(\mathbf{x})\} & =\left[\begin{array}{cc}
1 & (\mathbf{x}-\mathbf{c}_{0})^{\text{T}}\end{array}\right]\cdot\mathbb{E}\{\hat{\bm{\theta}}-\bar{\bm{\theta}}\}+e(\mathbf{x},\mathbf{c}_{0})\\
 & =\frac{1}{M}\sum_{m=1}^{M}e(\mathbf{x}_{m},\mathbf{c}_{0})+e(\mathbf{x},\mathbf{c}_{0})\\
 & +\frac{3}{Mr_{1}^{2}}\sum_{m=1}^{M}(\mathbf{x}-\mathbf{c}_{0})^{\text{T}}(\mathbf{x}_{m}-\mathbf{c}_{0})\cdot e(\mathbf{x}_{m},\mathbf{c}_{0})
\end{align*}

In addition, as $d(\mathbf{x}_{m},\mathbf{c}_{0})\leq r_{1}$ and
$d(\mathbf{x},\mathbf{c}_{0})\leq r_{0}$, the Lipschitz condition
implies that $|e(\mathbf{x}_{m},\mathbf{c}_{0})|\leq\frac{L_{g}^{2}}{2}d(\mathbf{x}_{m},\mathbf{c}_{0})^{2}\leq L_{g}r_{1}^{2}/2$
and $|e(\mathbf{x},\mathbf{c}_{0})|\leq L_{g}r_{0}^{2}/2$. The bias
is upper bounded as 
\begin{align}
\big|\mathbb{E}\{\hat{g}(\mathbf{x})-g(\mathbf{x})\}\big| & \leq\frac{L_{g}r_{1}^{2}}{2}+\frac{L_{g}r_{0}^{2}}{2}+\frac{3}{Mr_{1}^{2}}\cdot Mr_{0}r_{1}\cdot\frac{L_{g}r_{1}^{2}}{2}\nonumber \\
 & =\frac{L_{g}}{2}\left(r_{1}^{2}+r_{0}^{2}+3r_{0}r_{1}\right).\label{eq:bias-hat-gx}
\end{align}

Finally, substituting (\ref{eq:variance-hat-gx}) and (\ref{eq:bias-hat-gx})
into (\ref{eq:mse-g}), we have 
\begin{align}
\mathbb{E}\left\{ \left(\hat{g}(\mathbf{x})-g(\mathbf{x})\right)^{2}\right\}  & =\mathbb{V}\left\{ \hat{g}(\mathbf{x})\right\} +\left(\mathbb{E}\{\hat{g}(\mathbf{x})-g(\mathbf{x})\}\right)^{2}\nonumber \\
 & \leq\frac{\sigma^{2}}{M}\left(1+\frac{3r_{0}^{2}}{r_{1}^{2}}\right)\nonumber \\
 & +\frac{L_{g}^{2}}{4}\left(r_{1}^{2}+r_{0}^{2}+3r_{0}r_{1}\right)^{2}.\label{eq:mse-of-gx}
\end{align}

Furthermore, if $g(\mathbf{x})\approx g(\mathbf{c}_{0})+\nabla g(\mathbf{x})^{\text{T}}(\mathbf{x}-\mathbf{c}_{0})+L'_{g}\|\mathbf{x}-\mathbf{c}_{0}\|^{2}/2$
and $\|\mathbf{x}_{m}-\mathbf{c}_{0}\|^{2}=r_{1}^{2}$, then we have
$e(\mathbf{x},\mathbf{c}_{0})=L'_{g}\|\mathbf{x}-\mathbf{c}_{0}\|^{2}/2=L_{g}'r_{0}^{2}/2$
and $e(\mathbf{x}_{m},\mathbf{c}_{0})=L_{g}'r_{1}^{2}/2$ for all
$m$. The expression (\ref{eq:bias-hat-gx}) thus simplifies to $\mathbb{E}\{\hat{g}(\mathbf{x})-g(\mathbf{x})\}\approx L_{g}'\left(r_{1}^{2}+r_{0}^{2}\right)/2$.
Together with the result (\ref{eq:variance-hat-gx}), it leads to
the \ac{mse} approximation (\ref{eq:upper-bound-local-channel-model-1}).

\bibliographystyle{IEEEtran}

\newpage

\section*{Supplementary Materials}

\section{Proof of Proposition~\ref{prop:upper-bound-trajectory-length}}

\label{sec:Proof-of-Proposition-Trajectory-Length}

The superposed trajectory proposed in this paper combines the search
trajectory ${\bf x}_{\text{s}}(t)$ on the equipotential surface as
shown in Fig.~\ref{fig:Trajectory-on-the-Equipotential-Surface}
with the spiral measurement trajectory ${\bf x}_{\text{r}}(t)$ as
shown in Fig.~\ref{fig:sphere-and-cylinder}.

We first derive the trajectory length of the search trajectory ${\bf x}_{\text{s}}(t)$
on the equipotential surface. According to Proposition~\ref{prop:shape_of_equipotential_surface},
the equipotential surface is a sphere centered at $\mathbf{o}=(P_{0}\sum_{k\in\mathcal{K}}\mathbf{u}_{k}-P_{\text{T}}\mathbf{u}_{0})/(KP_{0}-P_{\text{T}})$
with radius $R$ in (\ref{eq:radius_2}). According to Section IV-B,
the search trajectory ${\bf x}_{\text{s}}(t)$ is categorized into
the following two phases:

Phase 1: trajectory ${\bf x}_{\text{s}}(t)$ in the full-LOS region.
The trajectory in the full-LOS region is an arc with a radius $R$
and a central angle less than $\pi$. Hence, the trajectory ${\bf x}_{\text{s}}(t)$
in the full-LOS region is upper bounded by $\pi R$.

Phase 2: trajectory in the non-full-LOS region. The trajectory in
the non-full-LOS region satisfies $f_{0}(g_{0}(\mathbf{x}_{\text{s}}(t)))=C$.
Therefore, ${\bf x}_{\text{s}}(t)$ in the non-full-LOS region is
an arc with a radius smaller than or equal to $H_{0}$ and a central
angle less than or equal to $\pi$. Then, the trajectory $\mathbf{x}_{\text{s}}(t)$
in the non-full-LOS region is upper bounded by $\pi H_{0}$.

By combining the trajectory length in phase 1 and phase 2, the length
of search trajectory $\mathbf{x}_{\text{s}}(t)$ is upper bounded
by $\pi(H_{0}+R)$.

Next, we derive the trajectory length of the spiral measurement trajectory
${\bf x}_{\text{r}}(t)$. The spiral measurement trajectory shown
in Fig.~\ref{fig:sphere-and-cylinder} (b) can be described by the
following 3D helix equations: 
\begin{equation}
\begin{cases}
x_{\text{r}1}(t)=\sqrt{2/3}r_{1}\cos(\omega t)\\
x_{\text{r}2}(t)=vt\\
x_{\text{r}3}(t)=\sqrt{2/3}r_{1}\sin(\omega t)
\end{cases}\label{eq:common-spiral-trajectory-equation-1}
\end{equation}
where recall that $\omega$ is the angular velocity, $v$ is the velocity
in the spiral axis direction, and $r_{1}$ is the measurement radius.

Therefore, the spiral trajectory length denoted as $L_{\text{r}}$
during $t\in[0,T]$ is given by
\begin{align}
L_{\text{r}} & =\int_{0}^{T}\sqrt{(-\sqrt{\frac{2}{3}}r_{1}\omega\sin(\omega t))^{2}+(\sqrt{\frac{2}{3}}r_{1}\omega\cos(\omega t))^{2}+v^{2}}\text{d}t\nonumber \\
 & =\sqrt{\frac{2r_{1}^{2}\omega^{2}}{3}+v^{2}}T.\label{eq:trajectory-length-spiral}
\end{align}

Finally, we combine the LOS discovery trajectory ${\bf x}_{\text{s}}(t)$
and the spiral measurement trajectory ${\bf x}_{\text{r}}(t)$. Given
the upper bound of ${\bf x}_{\text{s}}(t)$ as $\pi(H_{0}+R)$, the
total search time is upper bounded by $T\leq\pi(H_{0}+R)/v$. According
to the design of the spiral trajectory in Section III-E, let $\omega=4\pi/M$
and $v=2/\sqrt{M^{2}-1}$. Then, the total length $L$ of the superposed
trajectory is upper bounded by
\begin{align}
L & \leq\frac{\pi(H_{0}+R)}{v}\sqrt{\frac{2r_{1}^{2}\omega^{2}}{3}+v^{2}}\nonumber \\
 & =\pi(H_{0}+R)\sqrt{\frac{8(M^{2}-1)\pi^{2}r_{1}^{2}}{3M^{2}}+1}\nonumber \\
 & \leq\pi(H_{0}+R)\sqrt{3\pi^{2}r_{1}^{2}+1}.\label{eq:total-trajectory-length-bound}
\end{align}

\section{Per-step Computational Complexity Analysis}

\label{sec:Per-step-Computational-Complexity}

The computational complexity of the proposed algorithm mainly comes
from the local channel construction in Step 3, and the update of UAV
positions in Step 4 in Algorithm~\ref{alg:algorithm_equip}.
\begin{itemize}
\item Computational complexity of Step 3: The local channel construction
in Step 3 builds a local channel model for each user using the closed-form
formula in equation~(\ref{eq:least-square-solution}), and thus,
the computational complexity is given by $O(KM)$, where recall that
$K$ is the number of users and $M$ is the number of measurements
used for local channel model construction.
\item Computational complexity of Step 4: The update of \ac{uav} positions
in Step 4 has a computational complexity of $O((K+N)^{3})$, where
$N$ is the number of constraints for resource allocation. Step 4
depends on closed-form formulations in (\ref{eq:solution-to-local-approximation-equipotential-point}),
(\ref{eq:solution-dx-dt}), (\ref{eq:closed-form-q-in-LOS}), (\ref{eq:closed-form-q2-in-NLOS}),
(\ref{eq:Dynamical-Equation-Trajectory-x}), (\ref{eq:Dynamical-Equation-Trajectory-x-NoSwitch}),
(\ref{eq:Transition-Phase-x-s-dynamic}) and (\ref{eq:Dynamical-Equation-Trajectory-x-Switch-x}),
and the major complexity comes from the matrix inverse operation in
(\ref{eq:solution-dx-dt}).
\end{itemize}

\section{Proof of Lemma~\ref{lem:trace-inverse-positive-definite}}

\label{sec:Supplementary-Materials-2}

We first state the following two lemmas.
\begin{lem}[{Eigenvalue inequality \cite[Corollary 4.3.12]{RACR:B12}}]
\label{lem:eigenvalue-inequality} For a positive definite symmetric
matrix $\mathbf{B}\in M_{N}$ and $c\geq0$, we have
\[
0<\lambda_{i}(\mathbf{B})\le\lambda_{i}(\mathbf{B}+c),\quad i=1,2,\dots,N
\]
with the equality for some $i$ if and only if $c=0$.
\end{lem}
\begin{lem}[{Trace inequality \cite[Equation 2.4.2.1]{RACR:B12}}]
\label{lem:trace-inequalityr} For a positive definite symmetric
matrix $\mathbf{B}\in M_{N}$ and $c\geq0$, we have
\[
\text{tr}\left\{ \left(\mathbf{B}+c\right)^{-1}\right\} =\sum_{i=1}^{N}\lambda_{i}^{-1}(\mathbf{B}+c)\leq\sum_{i=1}^{N}\lambda_{i}^{-1}(\mathbf{B})=\text{tr}\left\{ \mathbf{B}^{-1}\right\} 
\]
with the equality if and only if $c=0$.
\end{lem}

Denote $\mathbf{A}_{N-i}\in M_{N-i}$ as a submatrix of $\mathbf{A}$
by removing the first $i$ rows and the first $i$ columns. Denote
$\mathbf{a}_{-i}$ as a column vector that contains the lower off-diagonal
elements of the $i$th column of $\mathbf{A}$. Thus,
\[
\mathbf{A}_{N}=\left[\begin{array}{cc}
a_{11} & \mathbf{a}_{-1}\\
\mathbf{a}_{-1}^{\text{T}} & \mathbf{A}_{N-1}
\end{array}\right],\,\mathbf{A}_{N-1}=\left[\begin{array}{cc}
a_{22} & \mathbf{a}_{-2}\\
\mathbf{a}_{-2}^{\text{T}} & \mathbf{A}_{N-2}
\end{array}\right].
\]

Using the block matrix inverse, we have 
\begin{align*}
 & \mathbf{A}_{N}^{-1}\\
 & =\left[\begin{array}{cc}
(a_{11}-\mathbf{a}_{-1}^{\text{T}}\mathbf{A}_{N-1}\mathbf{a}_{-1})^{-1} & *\\*
* & (\mathbf{A}_{N-1}-\mathbf{a}_{-1}^{\text{T}}a_{11}\mathbf{a}_{-1})^{-1}
\end{array}\right].
\end{align*}

Since $\mathbf{A}$ is positive definite, so are $a_{11}$ and $\mathbf{A}_{N-1}$.
Thus, $\mathbf{a}_{-1}^{\text{T}}\mathbf{A}_{N-1}\mathbf{a}_{-i}\geq0$
and $\mathbf{a}_{-1}^{\text{T}}a_{11}\mathbf{a}_{-1}\geq0$ for all
$\mathbf{a}_{-1}$ with equalities both at $\mathbf{a}_{-1}=\mathbf{0}$.

Thus, let $\mathbf{B}=\mathbf{A}_{N-1}-\mathbf{a}_{-1}^{\text{T}}a_{11}\mathbf{a}_{-1}$,
and $c=\mathbf{a}_{-1}^{\text{T}}a_{11}\mathbf{a}_{-1}$ where $\mathbf{B}$
must be positive definite because $\mathbf{A}_{N}$ is positive definite,
so as $\mathbf{A}_{N}^{-1}$ and $\mathbf{A}_{N-1}-\mathbf{a}_{-1}^{\text{T}}a_{11}\mathbf{a}_{-1}$.
Additionally, since $a_{11}>0$, $c\geq0$ with the equality if and
only if $\mathbf{a}_{-1}=\mathbf{0}$. Given Lemma~\ref{lem:eigenvalue-inequality}
and Lemma~\ref{lem:trace-inequalityr}, we have $\text{tr}\left\{ \mathbf{B}^{-1}\right\} \geq\text{tr}\{\left(\mathbf{B}+c\right)^{-1}\}$
with equality when $\mathbf{a}_{-1}=\mathbf{0}$. As a result, 
\begin{align*}
\text{tr}\{\mathbf{A}_{N}^{-1}\} & \geq a_{11}^{-1}+\text{tr}\{\left(\mathbf{A}_{N-1}-\mathbf{a}_{-1}^{\text{T}}a_{11}\mathbf{a}_{-1}\right)^{-1}\}\\
 & \geq a_{11}^{-1}+\text{tr}\{\mathbf{A}_{N-1}^{-1}\}
\end{align*}
with the equality when $\mathbf{a}_{-1}=\mathbf{0}$.

Similarly, we have
\begin{align*}
\text{tr}\{\mathbf{A}_{N-1}^{-1}\} & \geq a_{22}^{-1}+\text{tr}\{\mathbf{A}_{N-2}^{-1}\},\dots\\
 & \vdots\\
\text{tr}\{\mathbf{A}_{2}^{-1}\} & \geq a_{N-1,N-1}^{-1}+\text{tr}\{a_{N,N}^{-1}\}.
\end{align*}

Combining the above inequalities, one can obtain that $\text{tr}\{\mathbf{A}^{-1}\}\geq\sum_{i=1}^{N}a_{ii}^{-1}$
with the equality achieved for $a_{ij}=0$, $\forall i\neq j$.

\end{document}